\documentclass[format=acmsmall, review=True]{acmart}
\usepackage{acm-ec-25}
\usepackage{booktabs} 
\usepackage{xcolor}
\usepackage[ruled]{algorithm2e} 

\SetAlFnt{\small}
\SetAlCapFnt{\small}
\SetAlCapNameFnt{\small}
\SetAlCapHSkip{0pt}
\IncMargin{-\parindent}

\newcommand{\dram}{\texttt{DRAM}}

\usepackage{amsmath}
\DeclareMathOperator{\reg}{Reg}

\DeclareMathOperator{\tv}{TV}
\DeclareMathOperator{\lp}{LP}
\newcommand{\indep}{\perp\!\!\!\perp}

\usepackage{amsthm}  

\theoremstyle{remark}    
\newtheorem{remark}{Remark}

\usepackage{enumerate}

\setcitestyle{authoryear}

\newcommand{\blue}[1]{{#1}}

\title{Multi-agent Adaptive Mechanism Design}

\author{Qiushi Han}
\authornote{Author ordered alphabetically.}
\affiliation{%
  \institution{Massachusetts Institute of Technology}
  \country{USA},
  \textsf{joshhan@mit.edu}
}
\author{David Simchi-Levi}
\affiliation{%
  \institution{Massachusetts Institute of Technology}
  \country{USA},
  \textsf{dslevi@mit.edu}
}

\author{Renfei Tan}
\authornote{Corresponding author.}
\affiliation{%
  \institution{Massachusetts Institute of Technology}
  \country{USA},
  \textsf{rftan@mit.edu}
}

\author{Zishuo Zhao}
\affiliation{%
  \institution{National University of Singapore}
  \country{Singapore},
  \textsf{wiku30@nus.edu.sg}
}

\begin{abstract}
    We study the sequential mechanism design problem in which a principal seeks to elicit truthful reports from multiple rational agents while starting with no prior knowledge of agents’ beliefs.
    We introduce Distributionally Robust Adaptive Mechanism ($\dram$), a general framework combining insights from both mechanism design and online learning to jointly address truthfulness and cost-optimality.
    Throughout the sequential game, the mechanism would estimate agents' beliefs, then iteratively updates a distributionally robust linear program with shrinking ambiguity sets to reduce payments while preserving truthfulness. 
    Our mechanism guarantees truthful reporting with high probability while achieving $\tilde{O}(N\sqrt{T})$ cumulative regret, and we establish a matching lower bound showing that no feasible adaptive mechanism can asymptotically do better.
    The framework generalizes to plug-in estimators ($\dram +$), supporting structured priors and delayed feedback.
    To our knowledge, this is the first adaptive mechanism under the general settings that maintains truthfulness and achieves optimal regret when incentive constraints are unknown and must be learned.
\end{abstract}

\begin{document}






\maketitle

\section{Introduction}




The theory of mechanism design studies rules and institutions in various disciplines, ranging from auctions and online advertisements to business contracts and trading rules. 
The formulation often involves a central principal (system) and one or many rational agents (players), where the principal designs a mechanism to achieve a given objective subject to agents’ incentives. 
A typical component is the \textit{common knowledge assumption}: certain information about agents is presumed known to the principal and can be exploited to design analytically tractable, often optimal mechanisms. 
For example, knowledge over bidders’ value distributions over the auctioned can be used to design revenue-optimal auctions \cite{myerson1981optimal}. 
However, the availability of such knowledge is difficult to justify in practice. 
This observation, originally due to \cite{wilson1985game}, is now known as \textit{Wilson’s critique}. 
It proposes that some information is too private to be common knowledge, and such assumptions should be weakened to approximate reality.

In parallel, the theory of online learning studies algorithms that learn and make decisions in unfamiliar environments, aiming to approach the performance of oracles that have full knowledge from the start. 
The principal typically begins with no knowledge of the environment, and information is acquired through repeated data collection and carefully designed statistical methods. 
A common assumption is that the environment is unknown but \textit{stationary}. 
For example, in the classical multi-armed bandit model \cite{lattimore2020bandit}, each arm’s reward is a stochastic distribution, and the best arm can be discovered via repeated sampling.
An alternative is to assume the worst-case scenario from the environment, i.e., fully adversarial feedback.
These algorithms have wide applications in recommendation, pricing, scheduling, and more \cite{lattimore2020bandit}. 
In application, however, they often interact with humans, who are neither stationary nor fully adversarial. 
In fact, a foundational assumption in economics is that humans are \textit{rational} \cite{von2007theory}. 

Therefore, the strength and weaknesses from both fields seems to complement to each one.
Mechanism design incentivizes nice behavior from rational agents for proper learning guarantees, and online learning can provide the necessary knowledge for efficient mechanisms.
For this reason, the combination of mechanism design and online learning has received increasing attention, most notably in settings such as online contract design \cite{ho2014adaptive, zhu2022sample} and online auctions \cite{blum2004online, cesa2014regret}.
However, the design of general multi-agent adaptive mechanisms remains an under-explored problem.

In this work, we study the sequential mechanism design problem in which a principal in each round designs reward mechanisms for multiple rational agents, while starting with no prior knowledge of agents’ beliefs.
The principal's objective is three-fold: \textit{report quality}, \textit{truthfulness}, and \textit{cost-optimality}.
The principal wants to design a reward mechanism that can obtain the highest-quality data from the task, while incentivizing truthful report from agents, and do so in a cost-minimal way.
As a motivating example, consider the image labeling task where the principal assigns raw images to multiple agents for labeling (Figure \ref{fig: toy_example}).
In each round, each agent makes a private observation of the image as her type, then reports her type to the principal, and finally receives a payment.

\begin{figure}
    \centering
    \includegraphics[width=0.7\linewidth]{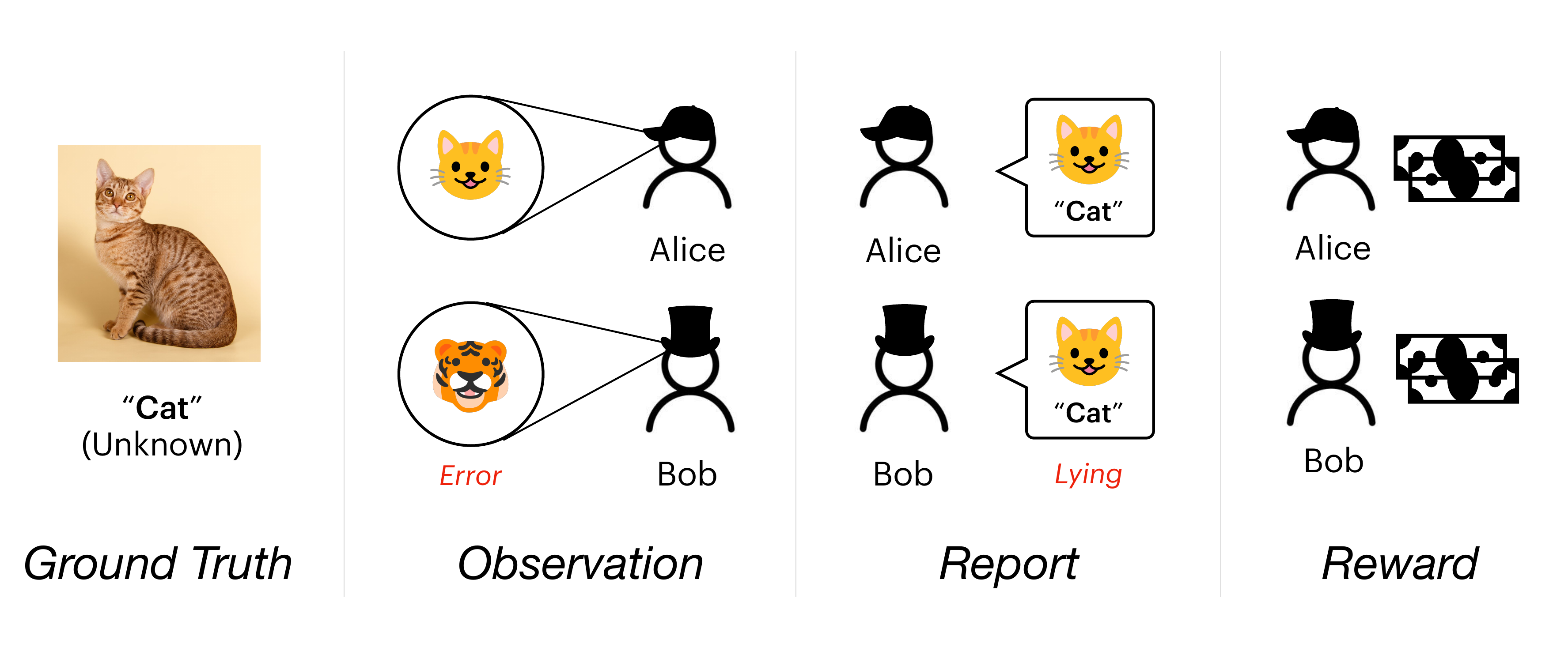}
    \caption{An image labeling example. Nature samples an unlabeled image with an unknown ground truth, which is then independently observed by multiple agents. Each agent's observation (type) is private to herself. The agents then report to the principal and receive rewards in the end. Lying or lazy behavior is possible, since the principal does not know the ground truth or the agents' observations. One objective is to incentivize truthful behavior via reward mechanisms based on only agents' reports.}
    \label{fig: toy_example}
\end{figure}

The multi-agent mechanism design problem faces several major challenges. 
First, each agent’s observation is private and unknown to the principal.
Agents are rational and pursue utility, so they may lie or become lazy.
However, ground truth may be unavailable or expensive to obtain, making it hard to directly control report quality or infer agent skills.
Second, classical mechanism designs that rely on common knowledge are inapplicable: even mechanisms that only focus on maintaining truthfulness often assume accurate knowledge of posteriors or correlation structure \cite{miller2005eliciting}. 
Finally, we cannot exploit the structures in specific mechanism design problems.
For example, in auction design, a second-price auction incentivizes truthfulness without any common knowledge.
However, for a general mechanism design problem, maintaining truthfulness itself under limited knowledge is already a substantial difficulty.


\subsection{Our Contribution}
Our work draws insights from both the mechanism design and online learning literature.
From a mechanism design perspective, we relaxed the common knowledge assumption, and generalizes the optimal mechanism design problem.
From an online learning perspective, we relaxed the setting from always-honest agents to the more realistic rational agents, and generalizes the prediction with expert advice problem.

\textbf{The necessity of truthfulness.}
We show that truthfulness is in a sense ``necessary'' for sequential decision making.
Any decision making process based on agents' rational reports can achieve highest performance if and only if agents are truthful (up to permutations).
This is a result based on Blackwell's informativeness theorem \cite{blackwell1953equivalent}, and stronger than the revelation principle \cite{myerson1979incentive}.
Since mechanism design is also a decision making task itself, it implies truthfulness is necessary for learning optimal mechanisms.

\textbf{Distributionally robust mechanisms.}
We introduce and study a family of distributionally robust mechanisms, which preserves truthfulness and incurs low cost even when the principal's knowledge is ambiguous.
We study the relations between design parameters and achievable robustness, establish methods to tractably acquire these mechanisms, and finally characterize the cost of robustness.

\textbf{Optimal adaptive mechanism design.}
We design a general framework named \textit{Distributionally Robust Adaptive Mechanism} ($\dram$), which estimate agents' beliefs and iteratively updates a distributionally robust mechanism based on estimation accuracy.
Our algorithm achieves an $\tilde{O}(N\sqrt{T})$ regret guarantee (up to logarithmic factors) while preserving truthfulness with high probability.
We complement this result with a matching lower bound showing that no feasible adaptive mechanism can do better in the worst case.
The theoretical results are validated by numerical simulations.
This framework extends to any plug-in estimators (e.g., structured or regularized estimators for discrete distributions) and is compatible with delayed or batched feedback. 
To our knowledge, this is the first general adaptive mechanism that maintains truthfulness and achieves optimal regret when incentive constraints depend on unknown and learned information.


\subsection{Related Work}

Our work draws insights from both the mechanism design and online learning literature.


\subsubsection{Online and adaptive Mechanism Design}


The combination of mechanism design and online learning is a fast-growing direction in algorithmic game theory \cite{roughgarden2010algorithmic}.
Prior work spans a variety of participation structures and problem settings.
Some works study \textit{synchronous} settings, where the same agents have multiple encounters with each other, as in repeated games \cite{hart2000simple, satchidanandan2023incentive, papadimitriou2022complexity}.
Others consider \textit{asynchronous} settings, where new agents may arrive and depart over time, a structure particularly relevant in online auctions and advertising platforms \cite{milgrom2019auction, choi2020online, hajiaghayi2004adaptive, cesa2007improved, wang2017display}. 
Under both settings, two domains have received the most attention: online contract design \cite{ho2014adaptive, zhu2022sample} and online auctions \cite{blum2004online, cesa2014regret}. 
These works typically focus on learning an optimal mechanism, such as an optimal contract or an optimal reserve price, using tools from bandit learning.

A key observation is that preserving truthfulness is not a substantive difficulty in these existing models. 
In contract design, agents’ choices naturally reflect their private information, and no truthful reporting constraint is involved. 
In online auctions, structural properties ensure incentive compatibility: for instance, in second-price auctions, truthful bidding remains a dominant strategy even when the reserve price is inaccurate.
Even without any knowledge, the second-price auction guarantee's agents' truthfulness.
As a result, works such as \cite{cesa2014regret} can safely explore suboptimal mechanisms during learning without risking incentive distortion or data contamination.
The focus is solely on learning an optimal mechanism, which makes these problems are more or less reducible to a bandit problem \cite{zhu2022sample, cesa2014regret}.

In contrast, in a general mechanism design problem, preserving truthfulness becomes a significant difficulty.
When the principal begins with ambiguous knowledge, an improperly constructed mechanism can immediately encourage agents to lie or exert low effort, thereby corrupting the collected data, undermining subsequent learning processes. 
Thus, unlike prior literature, maintaining truthfulness throughout the entire learning trajectory is not merely desirable but essential, and this requirement is one of the central challenge we must address.

\subsubsection{Prediction with Expert Advice}

We note that our setup is a stochastic, label-efficient variant of the classical \textit{prediction with expert advice} problem.
The prediction with expert advice problem is fundamental in online learning \cite{cesa2006prediction}.
In the standard framework, the agents (often called ``experts'') can sequentially provides arbitrary and even adversarial signals, and the principal's objective is to implement an \textit{aggregation} algorithm that achieves sublinear regret.
A simplification is to assume agents behave stochastically (report signals according to a probability law), under which the aggregation regret can be significantly improved \cite{cesa2004generalization}.
The stochastic variant also has connections with other online learning problems such as online optimization \cite{agarwal2017second, cesa2007improved, gaillard2014second}, with extensions in contextual or non-stationary settings \cite{besbes2016optimization}.

In the practical setting, acquiring a true label might be expensive.
It may be only feasible to query the true label for a small portion of rounds. 
This is the label-efficient setting of prediction with expert advice \cite{helmbold1997some, cesa2005minimizing}.
Roughly speaking, the aggregation regret decreases as the inverse square root of the number of queries.
There also exist adaptive algorithms that achieve the same regret with far fewer queries in benign cases \cite{mitra2020adaptivity, castro2023adaptive}.

Compared to the standard framework, our setup deals with \textit{rational} experts who need proper incentives for nice behaviors, which (on the difficulty of response aggregation) lies between the adversarial and the stochastic setting.
The assumption on rationality brings additional considerations on the design of incentives, which is the main concern of our work.
Also, different from the standard or the label-efficient settings, in our model, the true signal is never revealed (or only revealed for a constant number of rounds), adding difficulty to distinguish poorly-performed experts.

\subsubsection{Information elicitation and peer prediction}

The field of \emph{information elicitation} studies the mechanism design task to incentivize honest feedback from untrusted but rational participants, generally via designing \emph{scoring rules} \cite{li2022optimization} as rewards or penalties for participants. Particularly, \emph{peer prediction} \cite{miller2005eliciting} studies the scenarios in which ground truth is unavailable for direct verification of collected reports, with applications in dataset acquisition and evaluation \cite{chen2020truthful,zheng2024proper}, crowdsourcing \cite{dasgupta2013crowdsourced}, and recent blockchain-based decentralized ecosystems \cite{wang2023uncertainty,zhao2024takes}. The general paradigm of peer prediction mechanisms is to ask multiple participants the same question and reward them according to the comparison among their reports. While peer prediction mechanisms provide elegant results on truthful Nash equilibria without requirement of ground-truth information, most of existing mechanism rely on strong unrealistic assumptions on know prior and observation matrices, forming a gap to practical usage in real-world systems.

For practical usage, researcher develop a series of works with relaxed assumptions or stronger incentive guarantees. Particularly, \cite{kong2024dominantly} develops a prior-free multi-task peer prediction mechanism with dominant-strategy incentive compatibility, with a lack of permutation-proof property that is impossible for any prior-free peer prediction mechanisms \cite{kong2019information}. Besides, \cite{shnayder2016informed} provides an \emph{informed truthful} mechanism ensuring that the truthful equilibrium achieves highest utilities among all Nash equilibria, and \cite{zhang2025stochastically} develops a mechanism with a \emph{stochastic dominance} property ensuring incentive compatibility even under non-linear utilities.

In our setting, we address the gap between existing prior-dependent designs and reality via acquiring the prior distribution by \emph{online learning}, with a multi-round adaptive mechanism that learns the distributional information during the process. Besides, we also explicitly consider the robustness property that ensures incentive guarantees under inaccurate knowledge, thus making the peer prediction framework applicable in realistic applications.

\section{Problem Formulation}
We consider the sequential mechanism design problem where a principal seeks to elicit truthful reports from rational agents.
The principal sequentially assigns $T$ prediction tasks to a group of $N$ rational agents.
Each task has a true label $Y_t \in \mathcal{Y}$, i.i.d. sampled from an unknown and stationary distribution $p_Y(\cdot)$.
Unless stated otherwise, this true label is not revealed to anyone, either the principal or the agents.
We let $\mathcal{Y}$ be finite to avoid mathematical complications.
In each round, each agent $i= 1, \dotsb, n$ independently studies the task, acquiring her own observation $X_{it}\in \mathcal{Y}$ with a constant cost $c$.
$X_{it}$ is generated according to the agent's skill $p_i(x \mid y)$, a stationary conditional probability law.
Each agent might know her own skill distribution $p_i$, but has no information about anyone else’s, and the principal initially knows none of them.
Aside from observation, agents also has an outside option of lazily reporting a random label without observing the label (shirking, does not incur cost $c$ as well).

After studying the task, agents produce their \textit{public reports} $Z_{it} \in \mathcal{Y}$ to the principal.
\blue{Reports are independent across agents conditional on their observations and all history up to $t$.}
We assume agents are \textit{risk-neutral} and \textit{myopic}.
Being risk-neutral means agents aim to maximize their expected reward, conditional on the public and private information they have.
Being myopic means agents only care about immediate reward in the current round but not future rewards.
Under these settings, we have rational agents who do not necessarily report their observations.
Instead, they would \textit{lie} (report $Z_{it} \neq X_{it}$) or be \textit{lazy} (report $Z_{it}$ without observation $X_{it}$) when they expect an advantage in doing so.
We denote the observation and report profile of all agents in a round by $\boldsymbol{X}_t$ and $\boldsymbol{Z}_t$.
Note that $p_Y$ and $p_i$ together defines a joint law $p_{\boldsymbol{X}}$ over $\boldsymbol{X} \in\mathcal{Y}^N$, and we later show that learning this $p_{\boldsymbol{X}}$ is crucial for optimal mechanisms.

Collecting the reports, the principal rewards each agent $i$ with an reward mechanism $R_{it}(\boldsymbol{Z}_1, \dotsb, \boldsymbol{Z}_t)$.
The reports can then be used for downstream decision-making tasks, such as aggregation.
Note that the reward mechanism is non-anticipating, meaning the mechanism can only decide on past and current but not future report profiles.



The principal aims to design the online reward mechanism $\boldsymbol{R} = (R_{it})_{i\in [N], t\in [T]}$ with three objectives:
\begin{itemize}
    \item \textbf{Truthfulness} (aka \textit{incentive-compatibility}): given all other agents act honestly, a agent would maximize her own expected utility when she works, obtains observations, and then reports honestly ($Z_{it} = X_{it}$).
    
    \item \textbf{Report quality:} the reward mechanism should incentivize the highest-quality reports, such that downstream decision-making tasks may achieve the optimal objective.
    
    \item \textbf{Cost-optimality:} maintaining truthfulness and data-quality, the principal minimizes its total expected payment to agents.
\end{itemize}

We now compare between our setup and typical modeling assumptions in the online learning and mechanism design literature.
Online learning mainly targets at minimizing cumulative decisional error, while treating all reports as truthful. 
Mechanism design, by contrast, centers on strategic incentives, but usually presumes agents’ type distributions are known or even common knowledge. 
These assumptions ease analysis, yet rarely hold in practice. 
Our model pursues both goals at once and relaxes the assumptions from both fields. 


\begin{remark}
    \label{rmk: type report game}
    The proposed model can be further generalized to match the classical model in the mechanism design literature.
    Here, agents' observations are their own \textit{types} of the round.
    Assume in each round agents' types $\boldsymbol{X}$ are sampled from a stationary joint distribution.
    Agents then report their types (not necessarily truthful) $\boldsymbol{Z}$ to the principal and receive rewards.
    All of our analysis and algorithms applies to this generalized setting.
    In fact, our analysis does not make use of $p_Y$ and $p_i$ and focus exclusively on the joint law $p_{\boldsymbol{X}}$.
\end{remark}


\begin{remark}
    We note that each agent's utility is linear in only her own reward ($u_i = r_i$).
    In the most general setting of mechanism design, an agent's utility is a function of all agents' types and the resource allocation from principal: $u_i : \mathcal{Y}^N \times \mathcal{R} \to \mathbb{R}$, where $\mathcal{R}$ is the space of the principal's resource allocation decisions.
    For example, in contracts, utility depends on the agent's own type $x_i$ and the principal's payment $r_i$ ($u_i = f(x_i, r_i)$).
    In auctions, utility depends on whether or not the agent gets the item (with probability $p$), her valuation of the item (type $x_i$), and her requested payment $r_i$ ($u_i = p \cdot x_i - r_i$).
    Our analysis may potentially be generalized to the cases when the agents' utility function are not necessarily linear but still common knowledge.
\end{remark}

We conclude this section by revealing the importance of truthfulness.
After all, the principal's top objectives in outsourcing tasks are to improve data quality and lower costs.
Truthfulness, as a mechanism design objective, might not be of interest if the mechanism that reaches the highest quality or the lowest cost promotes dishonest behaviors.
From the revelation principle \cite{myerson1979incentive}, we know that truthfulness is ``free'', in the sense that we don't lose anything by focusing only on mechanisms with their incentive-compatible Nash equilibria.
For the same reason, it suffices to consider the setting where the true label $Y$, observation $X$, and report $Z$ all belong to the same space $\mathcal{Y}$.
The following proposition actually proves a stronger result, showing that truthfulness is not only free, but in fact almost \textit{necessary} for optimal downstream decision-making. 

\blue{Consider a decision-making problem with a finite action set $\mathcal{A}$ and bounded objective function $u:\mathcal{Y\times A \to \mathbb{R}}$ in a fixed round $t$.
Define the optimal value under a report profile $\boldsymbol{Z}$ by $V(\boldsymbol{Z}) = \max_\delta \mathbb{E}[u(Y_t, \delta(\boldsymbol{Z})]$, where the supremum ranges over all the deterministic decision rules $\delta$ measurable with respect to $\boldsymbol{Z}$.
Given a joint law $p_{Y, \boldsymbol{X}}$ on ground truth and agents' observations, the report profile is uniquely decided by the strategy profile of all agents.}

\begin{proposition}[permutative strategy achieves maximal report quality]
\blue{
\label{prop: truthfulness is necessary}
Consider a fixed round $t$ in the sequential mechanism design problem. Let the participating agents fix a strategy profile. Conditional on any public history up to $t-1$, the following statements are equivalent:
\begin{itemize}
    \item The strategy profile produces a report $\boldsymbol{Z}_t$ that attains the maximal quality across all possible $\boldsymbol{Z}$:
    \begin{align*}
        V(\boldsymbol{Z}_t) \geq V(\boldsymbol{Z}), \quad \text{for every joint law $p_{Y_t, \boldsymbol{X}_t}$ and bounded decision problem $(\mathcal{A}, u)$}.
    \end{align*}
    \item All agents observes and adopt a permutative reporting strategy: for each agent $i$, there exists a bijection $\pi_i:\mathcal{Y} \to \mathcal{Y}$ such that \(\mathbb{P}(Z_{it} = \pi_i(X_{it})) = 1.\)
\end{itemize}}

\end{proposition}


Proposition \ref{prop: truthfulness is necessary} is derived from Blackwell's informativeness theorem \cite{blackwell1953equivalent}.
Each round, the true label, observations, reports, and decisions form a Markov chain: $Y_t \to \boldsymbol{X}_t \to \boldsymbol{Z}_t \to A_t$.
An intuition is that optimal decision-making requires maximal \textit{information} from upstream.
Due to the data processing inequality \cite{cover1999elements}, information never increases going downstream; therefore, the best approach is to preserve as much information as possible at each link.
Truthful reporting preserves full information at the link $\boldsymbol{X}_t \to \boldsymbol{Z}_t$, which allows for optimal subsequent decisions.
Any lies that mixes up the labels would erode information.
Lazy behavior also produces less information than observing.
Aside from truthfulness, an alternative case that preserves full information is when agents permutes the symbol before reporting.
However, such a case is unrealistic in practical settings, as the principal would need to know each agent's permutation rule to reverse the encoding and uncover the true observation.
Therefore, this proposition essentially shows that eliciting truthfulness is the practical way to achieve maximal report quality.


Proposition \ref{prop: truthfulness is necessary} also implies that truthfulness is crucial not only for report quality but also cost-optimality.
This is because the design of cost-optimal mechanisms is itself a downstream decision-making task, and truthfulness ensures that optimal mechanisms can be achieved.

\section{Mechanism Design without Common Knowledge}
\label{sec: robust mechanism design}

In our work, a central relaxation of modeling is that we don't assume prior distribution of labels $p_Y$ or agents' skills $p_i$ are known by agents or the principal.
The principal's attempts to maintain truthfulness with unknown or inaccurate estimation of such knowledge.
In this section, we focus on distributionally robust mechanisms, which aim to incentivize truthful behavior under knowledge ambiguity.

\subsection{Optimal Single-round Mechanism Design}

We begin with the analysis of optimal mechanism design with \textit{known} $p_Y$ and $p_i$ within a single round.
When there are no true labels available, we apply the principles of peer prediction, which is to use other agents' report to verify a focal agent's report.
The delicacy lies in the careful design of the reward mechanism to ensure that truthfulness is a Nash Equilibrium.


We start with the two-agent mechanism. 
With a focal agent $i$ and a reference agent $j$, the optimal two-agent mechanism design problem could be formulated as a linear programming problem. 
The objective is to minimize expected reward to agents, and the constraints are the desired properties of the mechanism.
(we hide subscript $t$ for simplicity.)
\begin{equation}
    \label{eq: two-agent mechanism design}
    \begin{split}
        \min_{R_i} \quad& \mathbb{E}[R_i (X_i, X_j)] \\ 
        \textrm{s.t.} 
        \quad& \mathbb{E}[R_i (X_i, X_j)\mid X_i] \geq c \\
        \quad& \mathbb{E}[R_i (Z_i, X_j)\mid X_i] \leq c, \quad \forall Z_i \neq X_i\\
        \quad& \mathbb{E}[R_i (Z_i, X_j)] \leq 0, \quad \forall Z_i \indep X_i, X_j 
    \end{split}
\end{equation}

The expectation is taken under the joint probability law $p_{\boldsymbol{X}}$ induced by $p_Y, p_i,$ and $p_j$. 
In fact, we would show that $p_{\boldsymbol{X}}$ is the core parameter of the mechanism design problem.
\blue{The first constraint enforces the \textit{individual rationality} property, ensuring an observant, truthful agent obtains non-negative expected reward.
The second constraint states that any lying behavior receives a non-positive expected reward.
Together with the first constraint, they imply \textit{truthfulness} property, but in a stronger form whereas truthfulness provides at least $c$ and lying provides at most $c$.
The final constraint implements the \textit{no-free-lunch} property, meaning a lazy agent cannot get positive expected reward.}
We introduce the individual rationality and the no-free-lunch constraint as normalization to prevent arbitrary decrease of the objective function under affine transformations to rewards. (For risk-neutral agents, affine transformations on reward do not affect utility ordering and strategic behavior.)
\blue{In the following, we call a mechanism \textit{feasible} for player $i$ if it satisfies the three constraints in Eq.\eqref{eq: LP formation for two-agent mechanism design}.}

\begin{example}[Image Labeling]
    \label{eg: image labeling}
    Suppose there are two types of images $\mathcal{Y} = \{\texttt{Cat}, \texttt{Tiger}\}$, abbreviated with $C$ and $T$ respectively.
    We further assume that the prior distribution of the image types is balanced, i.e., $p_Y(C) = p_Y(T) = 0.5$.
    For each image with an unknown true label $Y \in \mathcal{Y}$, the principal would like to let two agents $1, 2$ individually observe it, and truthfully report their observations $X_1, X_2$ to label that image. 
    Assume that both agents are $90\%$ accurate: $p_i(C \mid C) = p_i(T\mid T) = 0.9$, and $p_i(T \mid C) = p_i(C \mid T) = 0.1$.
    
    The principal designs the reward mechanism as follows: both agents receive $1$ reward if their reports $Z_1, Z_2$ agree with each other, and receive $-1$ otherwise, i.e. $R_{\{1,2\}}(Z_1, Z_2) = 2\cdot \mathbf{1}_{[Z_1=Z_2]} - 1$. and assume that observation incurs cost $c=0.1$.
    Now we assume that agent $2$ observes and report honestly, and analyze the incentive of agent $1$.

    Suppose agent $1$ observes \texttt{Cat}, Bayes' formula gives that $\mathbb{P}(X_2 = C\mid X_1 = C) = 0.82$ and $\mathbb{P}(X_2 = C\mid X_1 = C) = 0.18$.
    On the other hand, if agent $1$ does not pay the effort to toss the coin, then her Bayesian belief on agent $2$'s observation (and report) is $P(X_2=C)=P(X_2=T)=0.5$.
    She can then work out expected reward under truthful, lying, and lazy strategies: truthful ($0.54$) > lazy ($0$) > lying ($-0.74$).
    Hence truthful behavior is desired.
    In fact, this simple mechanism is a feasible solution to Eq.\eqref{eq: two-agent mechanism design}.
    Intuitively, after the focal agent's observation, her posterior probability on the other agent's observing the same label is higher than observing a different label. 
    Therefore, it is preferable to report whatever you observe in the first place.
    Such a mechanism is called \textit{peer prediction} \cite{miller2005eliciting}, originated the fact that rational agents always tries to predict their peers' observations before action.
\end{example}

Define the belief matrix $\mathbf{B}$ where $\mathbf{B}_{xx'} = \mathbb{P}(X_j = x'\mid X_i = x)$, and reward matrix $\mathbf{R}$ where $\mathbf{R}_{xx'} = R_i(x, x')$.
We also let $\mathbf{d}$ be a column vector display of $j$'s observation distribution: $\mathbf{d}_x = \mathbb{P}(X_j = x)$.
Then we can reformulate \eqref{eq: two-agent mechanism design} into the following equivalent problem:
\begin{equation}
    \label{eq: LP formation for two-agent mechanism design}
    \begin{split}
        \min_{\mathbf{R}} \quad& \sum_{x,x'}\mathbb{P}(X_i = x)  \mathbf{B}_{xx'}\mathbf{R}_{xx'} \\ 
        \textrm{s.t.} 
        \quad& (\mathbf{B}\mathbf{R}^\intercal)_{xx} \geq c, \quad \forall x\in \mathcal{Y}\\
        \quad& (\mathbf{B}\mathbf{R}^\intercal)_{xy} \leq c, \quad \forall x\neq y \in \mathcal{Y}\\
        \quad&  \mathbf{R}\mathbf{d} \leq \mathbf{0}
    \end{split}
\end{equation}

Note that the second constraint only enforces \textit{pure} lying strategies to incur non-positive reward, nevertheless, it is sufficient since any mixed strategy is a convex combination of pure strategies and its corresponding reward is also a convex combination with the same weights.
The final constraint assumes all entries of $\mathbf{R}\mathbf{d}$ are negative, thus making sure any report strategy without observing incurs a non-positive reward.


\begin{theorem}[Optimal cost of a two–agent peer-prediction mechanism]
\label{thm: optimal cost of 2-agent peer prediction}
\blue{Assume that $\mathbb{P}(X_i = x) > 0$ for all $x\in \mathcal{Y}$, and the belief matrix $\mathbf{B}$ is invertible.}
Then the linear program \eqref{eq: two-agent mechanism design} (equivalently, its matrix form \eqref{eq: LP formation for two-agent mechanism design}) is feasible.  
Moreover, its optimal value equals the labor cost~$c$; that is,
\[
\min_{R_i\,\text{satisfying}\;\eqref{eq: two-agent mechanism design}}
\;  \mathbb{E}[R_i (X_i, X_j)] 
\;=\; c .
\]
In addition, at optimality, the first constraint in \eqref{eq: two-agent mechanism design} is binding.
\end{theorem}

The tight result on the objective function is in the spirit of the classical Cr\'emer-McLean mechanism \cite{cremer1988full}, which can extract full surplus from the agents when type distributions are common knowledge.
The conditions are satisfied for ``almost all'' $\mathbf{B}$ and $\mathbf{d}$.

\begin{example}[Optimal Mechanism in Image Labeling]
    \label{eg: optimal mechanism in image labeling}
    Continuing the image labeling example from Example \ref{eg: image labeling}, we show that the optimal mechanism pays both agents the observation cost $c=0.1$ in expectation, as a demonstration of Theorem \ref{thm: optimal cost of 2-agent peer prediction}.
    The optimal mechanism can be acquired by solving Eq.\eqref{eq: LP formation for two-agent mechanism design}.
    Here, we first compute the belief matrix and agent $2$'s observation distribution.
    \[
        \mathbf{B} = 
        \begin{bmatrix}
            0.82 & 0.18 \\
            0.18 & 0.82
        \end{bmatrix},\quad
        \mathbf{d} = 
        \begin{bmatrix}
            0.5 \\
            0.5
        \end{bmatrix}.
    \]
    Solving the linear program would give the following mechanism: both agents receive $5/32$ reward if their reports agree, and receive $-5/32$ otherwise.

    We now show that the mechanism satisfies all constraints and is cost-optimal.
    First, no-free-lunch is satisfied since when a agent reports without observation, no matter what strategy she follows, expected reward is always $0.5 \times 5/32 - 0.5 \times 5/32= 0$ (since the other agent observes both labels with equal probability).
    When she lies, expected reward is $(\mathbf{B}\mathbf{R}^\intercal)_{xy} = -0.1$.
    When she is truthful, she gets $\sum_{x,x'}\mathbb{P}(X_i = x)  \mathbf{B}_{xx'}\mathbf{R}_{xx'} = 0.1$ in reward, exactly equal to her observation cost.
    The optimal mechanism extracts full surplus from agents.
\end{example}

\begin{remark}
    From Lemma 1 of \cite{radanovic2013robust}, it is known that for any mechanism with more than two agents with a truthful Bayesian Nash Equilibrium, it is possible to construct a 2-agent mechanism (with one focal and one reference agent) that achieves a truthful Bayesian Nash Equilibrium with the same expected payment.
    This lemma narrows our attention to reward mechanisms with only two agents.
    In application, the reference agent can be randomly picked to avoid collusion.
\end{remark}

\subsection{Distributionally Robust Mechanism Design with Inaccurate Knowledge}

Now we move on to the scenario where agents and the principal only have \textit{inaccurate} knowledge of $p_{\boldsymbol{X}}$.
Suppose they only know that the true distributions agents' observations $\boldsymbol{X}$ belong to some ambiguity sets $p\in \mathcal{P_{\boldsymbol{X}}}$.
From the principal's perspective, the challenge would be to design a \textit{distributionally robust} mechanism, such that for any possible realizations within the ambiguity set, the truthfulness constraint would still be met.
Its objective now becomes minimizing the expected payment in the worst case.

The notion of (distributionally) robust mechanisms has been studied in \cite{bergemann2005robust, koccyiugit2020distributionally}.
We focus on a specific family that is cheap enough and guarantees truthfulness under ambiguity, but we don't pursue exact worst-case optimality.
This greatly reduces computational complexity, and turns out to be sufficient for subsequent adaptive mechanism design.



We begin with a sensitivity analysis of \eqref{eq: two-agent mechanism design} with respect to shifts in probability law.
Assume the principal obtain a design by solving \eqref{eq: two-agent mechanism design} according to an erroneous probability law $p$, but the true law is $p^*$.
According to Theorem \ref{thm: optimal cost of 2-agent peer prediction}, at optimality, we have a binding constraint in \eqref{eq: two-agent mechanism design}.
Therefore, any slight deviation from $p$ would lead to the potential violation of constraints.
To hedge against violations, following the idea of \cite{zhao2024takes}, the principal could add \textit{safety margins} on the constraints.
Instead of only requiring the expected reward of truthful behaviors to be greater than $c$, the principal could let it be no less than $c+\delta$, where $\delta >0$ is the margin width.
In this case, even if the expected reward of truthful behaviors might decrease under $p'$, as long as the decrease is no more than $\delta$, the individual rationality property is still preserved.
Under this idea, we look at a variant of the mechanism design problem with margin $\delta$:
\begin{equation}
    \label{eq: robust two-agent mechanism design}
    \begin{split}
        \min_{R_i} \quad& \kappa \\ 
        \textrm{s.t.} 
        \quad& |R_i| \leq \kappa \\
        \quad& \mathbb{E}_p[R_i (X_i, X_j)\mid X_i] \geq c+\delta \\
        \quad& \mathbb{E}_p[R_i (Z_i, X_j)\mid X_i] \leq c-\delta, \quad \forall Z_i \neq X_i\\
        \quad& \mathbb{E}_p[R_i (Z_i, X_j)] \leq -\delta, \quad \forall Z_i \indep X_i, X_j 
    \end{split}
\end{equation}


There are two curious features to this problem variant.
First, the margin $\delta$ added to the constraints protects the principal against inaccurate knowledge at an additional cost of at least $\delta$, since the expected payment under each possible observation is at least $c+\delta$.
This is a lower bound on the cost of robustness.
Increasing $\delta$ means the mechanism from \eqref{eq: robust two-agent mechanism design} is robust for a higher degree of inaccuracies, but it would also cost more.
Pursuing optimality, the principal want to find the lowest $\delta$ just enough to guarantee constraints are still satisfied under $p^*$.
It would be crucial to understand the connection between the degree of misspecification and the minimal required margin $\delta$.
Second, the objective function alters from minimizing expected payment to minimizing worst-case payment.
The reason for this change is to limit the sensitivity of expected payment to worst case probability law deviation from $p^*$ to $p$.
Following the ``compactness'' criteria discussed in \cite{zhao2024takes}, the outcome incurring the highest absolute payment has the highest sensitivity to probability deviation, hence a large $\delta/\kappa$ ratio would ensure a high robustness to such deviations. Therefore, for a fixed $\delta$, we would like to lower $\kappa$ as much as possible.

\begin{theorem}[Robustness to distributional misspecification]
\label{thm: robustness}
\blue{Let $p$ and $p^*$ be two joint distributions of $(X_i, X_j) \in \mathcal{Y}$.
Let $R_i: \mathcal{Y}\times \mathcal{Y} \to \mathbb{R}$ be a mechanism feasible for the margin-$\delta$ constraints (Eq.\eqref{eq: robust two-agent mechanism design}) under $p$ and $\delta > 0$.
Assume that $p(X_i = x) > 0$ and $p^*(X_i = x) > 0$ for all $x \in \mathcal{Y}$.
}
 
Denote the total variation distance by $\tv(p,p^*)$. Suppose that  
\begin{align}
    \label{eq: tv condition}
    \tv\left(p(\cdot \mid X_i = x),p^*(\cdot \mid X_i = x)\right) \leq \frac{\delta}{2\kappa}, \quad \forall x\in \mathcal{Y} \cup \{\varnothing\},
\end{align}
where $\kappa = \max |R_i|$ and we use $p(\cdot \mid X_i = \varnothing$) to denote unconditional laws.
Then the same mechanism $R_i$ produced remains feasible for the original problem \eqref{eq: two-agent mechanism design} when evaluated under $p^*$.

\end{theorem}

 


Theorem \ref{thm: robustness} proves the robustness of the margin-$\delta$ mechanism under inaccurate distributional knowledge.
This suggests it is possible to design a mechanism that guarantees truthfulness for all distributions in an ambiguity set centered at a distribution estimation $p$: $\mathcal{P} = \{p'| p' \,\text{satisfying}\;\eqref{eq: tv condition}\}$.
However, notice that the objective $\kappa$ itself is influenced by the margin $\delta$ we choose and the distribution $p$ used.
Actually, when we increase $\delta$, the minimal achievable $\kappa$ would also increase, hence the increment in the provided robustness  (i.e. $\delta/\kappa$) may diminish.
Therefore, one cannot infinitely increase $\delta$ hoping for unlimited robustness.
In the following, we first provide an objective upper bound on \eqref{eq: robust two-agent mechanism design}, then provide the upper and lower bounds on the amount of robustness that can be provided by \eqref{eq: robust two-agent mechanism design}.




\begin{theorem}[Bounds on payments of robust mechanism]
    \label{thm: bounds on payments of robust mechanism}
    \blue{Fix a design distribution $p$ and margin $\delta$. Assume that $p(X_i = x) > 0$ for all $x\in \mathcal{Y}$, and the belief matrix $\mathbf{B}$ induced by $p$ is invertible.}
    Let $\kappa^*$ be the optimal value of \eqref{eq: robust two-agent mechanism design} under $p$ and $\delta$.
    Then  we have
    \begin{itemize}
        \item Worst-case payment:
            \begin{equation}
                \label{eq: kappa upper bound}
                c+\delta \; \leq\; \;\kappa^* \;\leq\; \|\mathbf{B}^{-1}\|_2 \cdot \left(c\cdot\frac{\gamma|\mathcal{Y}| + 1}{1-\gamma} + \delta \cdot\frac{(1+\gamma)|\mathcal{Y}|+2}{1-\gamma}\right).
            \end{equation}
            where $\gamma = \max_x p(X_i = x) < 1$.
        \item Expected payment: a feasible solution $(\kappa, R_i)$ could be constructed that satisfies the bound \eqref{eq: kappa upper bound}, while ensuring the expected reward in the truthful equilibrium under $p$ is $c+\delta$.
    \end{itemize}
\end{theorem}

The essential insight from Theorem \ref{thm: bounds on payments of robust mechanism} is that \eqref{eq: robust two-agent mechanism design} is a linear programming problem.
Hence, if we consider $\delta$ as a perturbation on the constraints, the shift in the objective itself should also linearly relate to the perturbation.
In other words, the sensitivity of both worst-case and expected payment to perturbation $\delta$ is $O(\delta)$.

\begin{corollary}[Bounds on achievable robustness]
\label{col: robustness-floor}
Let the (possibly inaccurate) design distribution be $p$.
Under the conditions of Theorem \ref{thm: bounds on payments of robust mechanism}, for any margin parameter $\delta > 0$, there exists a mechanism $R_i$ feasible for
\eqref{eq: robust two-agent mechanism design} such that
\begin{equation}
    \delta/2\kappa\;\ge\; \frac{\delta \cdot (1-\gamma)}{2\|\mathbf{B}^{-1}\|_2 (c\cdot (\gamma|\mathcal{Y}| + 1) + \delta \cdot ((1+\gamma)|\mathcal{Y}|+2))},
\end{equation}
where $\kappa=\max |R_i|$ and $\gamma = \max_x \mathbb{P}(X_i = x) < 1$.

Moreover, we have $\delta/2\kappa \leq 1/2$ for all $\delta$ and corresponding feasible mechanism $R_i$, meaning no mechanism has its robustness be more than $1/2$.
\end{corollary}

Corollary \ref{col: robustness-floor} is obtained by simply replacing $\kappa$ in $\delta/2\kappa$ with its upper bound in Eq.\eqref{eq: kappa upper bound}.
It provides the minimum robustness achieved by solving \eqref{eq: robust two-agent mechanism design}.
Notice that increasing $\delta$ provides more robustness, but comes at increased cost.
Also, the marginal robustness from increasing $\delta$ would diminish: as $\delta \to \infty$, the lower bound is at most $(1-\gamma)/2\|\mathbf{B}^{-1}\|((1+\gamma)|\mathcal{Y}|+2)$.
When $\delta\to 0$, the robustness provided scales linearly with $\delta$.
In addition, from the lower bound on the worst case payments, we have an upper bound $1/2$ on the maximum robustness that can be provided from our scheme.
This result does not contradict the impossibility result (Theorem 1) shown in \cite{radanovic2013robust}, which proves that no mechanism can guarantee truthfulness for \textit{all} distributions.

Our theorem shows that while an all-round strictly dominantly truthful mechanism does not exist, it is still possible to design a mechanism that covers distributions that are relatively close to a design distribution $p$.
For distributions distanced too far away, the impossibility result still holds.
In other words, our scheme is enough to cover inaccuracies that are not too extreme.
For that reason, Corollary \ref{col: robustness-floor} is already sufficient for our purposes of adaptive mechanism design.
If the principal starts with a distribution estimation not too off (such that the total variation distance condition is satisfied with the robustness floor), then it is possible to maintain truthfulness and refine estimation at the same time.
The principal would first apply a mechanism that provides abundant robustness for its initial ambiguity set.
As more data is obtained and estimation becomes more accurate, it shrinks the ambiguity set and selects a smaller $\delta$, eventually converging to $\delta=0$ and the optimal mechanism design. 

\begin{theorem}[Cost of robustness]
    \label{thm: cost of robustness}
    \blue{Fix a design distribution $p$. Assume the conditions of Theorem \ref{thm: bounds on payments of robust mechanism}, and define $\eta \leq \tilde{\eta} = (1-\gamma)/2\|\mathbf{B}^{-1}\|_2((1+\gamma)|\mathcal{Y}| +2)$ with $\gamma = \max_x p(X_i = x) < 1$.}
    Let $\eta \in [0, \tilde{\eta})$. Consider the ambiguity set 
    \begin{align*}
        \big\{ p'\in\mathcal{P} \; \big| \; \tv\left(p(\cdot \mid X_i = x),p' (\cdot \mid X_i = x)\right)  \leq \eta,\;\forall x\in \mathcal{Y} \cup \{\varnothing\} \big\}.
    \end{align*}
    Then there exists a mechanism $R_i$ such that truthful reporting is a best response for agent $i$ if her belief belongs to this ambiguity set.
    Moreover, if the true distribution $p*$ belongs to this set, then the principal's expected payment under truthful reporting is at most
    \begin{align}
        \label{eq: cost of robustness}
        c +  \frac{4\|\mathbf{B}^{-1}\|_2 (\gamma|\mathcal{Y}| + 1)\cdot\eta}{(1-\gamma) - 2\|\mathbf{B}^{-1}\|_2((1+\gamma)|\mathcal{Y}| +2)\cdot \eta}\cdot c, 
    \end{align}
    where the second term is the additional cost of robustness.
    
\end{theorem}

Theorem \ref{thm: cost of robustness} is essentially a combination of Theorem \ref{thm: robustness} and \ref{thm: bounds on payments of robust mechanism}.
Notice that the additional cost takes the format of $c\cdot C_1\eta/(1-C_2\eta)$.
This means when $\eta \to 0$, the additional cost of robustness is roughly $O(\eta)$, and we establish a linear relation between robustness and additional cost.
\blue{Also, this robust mechanism can be obtained simply by solving a linear equation (See proof of Theorem \ref{thm: bounds on payments of robust mechanism} and \ref{thm: cost of robustness}).
We would use this construction method repeatedly in our adaptive mechanism.
}

\begin{example}[Distributionally Robust Mechanism in Image Labeling]
    \label{eg: dist robust mech image labeling}
    We follow the same example as in Example \ref{eg: image labeling} and \ref{eg: optimal mechanism in image labeling}.
    Now we compare the simple mechanism that pays $1$ on agreement and the optimal mechanism that pays $5/32$.
    Although the simple mechanism is suboptimal, it is robust to misspecification of agents' skills.

    For example, suppose that the agents' true observation accuracy is $0.8$ instead of $0.9$.
    With the same procedure in Example \ref{eg: image labeling}, one can show that the simple mechanism still guarantees truthfulness (truthful ($0.26$) > lazy ($0$) > lying ($-0.46$).
    On the other hand, the previously optimal mechanism breaks down (truthfulness gives $9/160 < c = 0.1$, so agents have no incentives in participation).
    In fact, one can show that as long as both agents' accuracies are the same and stay within the range $[0.66, 1]$, the simple mechanism always guarantees truthfulness.
    The lower bound $(10+\sqrt{10})/20\approx 0.66$ is when the truthful strategy's expected reward falls to $0.1$.
    This property holds true even if agents know the actual skill level, while the principal does not have that information.
    In a word, the additional payment in the simple mechanism serves as insurance against ambiguity.
\end{example}

\section{Adaptive Mechanism Design}
\label{sec: adaptive mech design}

In this section, we study the problem of adaptive mechanism design where the principal has no initial knowledge.
Define the principal's (empirical cumulative) \textit{regret} after $T$ rounds as $\sum_{t=1}^T \sum_{i=1}^N (R_{it} - c)$, where $c$ is the optimal payment that could guarantee feasibility shown by Theorem \ref{thm: optimal cost of 2-agent peer prediction}.
Starting from oblivion, the principal's aim is to find an adaptive mechanism that guarantees high probability truthfulness while minimizing regret.

We present our algorithm, ``Distributionally Robust Adaptive Mechanism'' (\texttt{DRAM}), in Algorithm \ref{alg: dram}.
The algorithm maintains truthfulness and reduces cost by designing a sequence of distributionally robust mechanisms with shrinking ambiguity parameter $\eta$.
The ambiguity parameter tracks the accuracy of the principal's estimation at each round.

\begin{algorithm}[h]
\caption{Distributionally Robust Adaptive Mechanism (\texttt{DRAM})}
\label{alg: dram}
\DontPrintSemicolon
\SetKwInOut{Input}{Input}
\Input{ambiguity threshold $\tilde{\eta}$; failure tolerance $\varepsilon$; lower bound on observation frequency $0< \rho < \min_{i, x\in\mathcal{Y}}\mathbb{P}(X_i = x)$; upper bound on observations $\max_y \mathbb{P}(X_i = y) < \gamma_{(i)} < 1$.}

Update ambiguity threshold $\tilde{\eta} = \min (\tilde{\eta}, 1/\sqrt{2})$;

Compute the warm-start phase length $\tau = \log((d+1)2^d N \log T /\varepsilon)/2\rho\tilde{\eta}^2$;

For each agent $i$, assign a corresponding reference agent $j$.\;

\BlankLine
\textbf{\underline{Warm-start phase.}}\;
\For{$t=1,2,\ldots, \tau$}{
    Obtain true label $Y_t$ from an external source.\;
    Deploy the fact-checking mechanism $R_{i} = s \mathbf{1}\{Z_{it} = Y_t\}$ with sufficiently large $s$ for each agent $i$.\;
}

\BlankLine
\textbf{\underline{Adaptive phase.}}\;
Define epoch schedule $\tau = \tau_0<\tau_1<\tau_2<\cdots $ with $\tau_k-\tau_{k-1}=2^{\,k-1}\tau$ \textit{(continue until $\tau_m\!\ge\!T$ or anytime)}.\;

\For{$k=1,2,\ldots,m$}{
    Estimate reference distribution with all past reports:
    \[\hat{p}_{ik}(x_j \mid x_i) = \frac{\sum_{t=1}^{\tau_{k-1}} \mathbf{1} \{Z_{it} = x_i, Z_{jt} = x_j\}}{\sum_{t=1}^{\tau_{k-1}} \mathbf{1}\{Z_{it} = x_i\}}, \quad \quad \hat{p}_{ik}(x_j \mid \varnothing) = \frac{1}{\tau_{k-1}}\sum_{t=1}^{\tau_{k-1}} \mathbf{1}\{Z_{jt} = x_j\}.\]
    
    Let ambiguity parameter $\eta_k = \sqrt{\log((d+1)2^dNm/\varepsilon)/2\rho\tau_{k-1}}$.
    
    For each agent $i$, set their safety margin \[\delta_{ik} = \frac{2\|\mathbf{B}^{-1}_{(i)}\|_2 (\gamma_{(i)}|\mathcal{Y}| + 1)\cdot\eta_k}{(1-\gamma_{(i)}) - 2\|\mathbf{B}^{-1}_{(i)}\|_2((1+\gamma_{(i)})|\mathcal{Y}| +2)\cdot \eta_k}\cdot c.\]
    Here $\mathbf{B}_{(i)}$ is the matrix representation of $\hat{p}_{ik}(x_j \mid x_i)$ (see Section \ref{sec: robust mechanism design}), and $\gamma_{(i)} = \max_x \mathbb{P}(X_i = x)$.
    
    \blue{Construct the robust mechanism $R_{ik}$ for each agent Eq.\eqref{eq: robust two-agent mechanism design} with parameter $\hat{p}_{it}$ and $\delta_{ik}$}.
    
    Deploy the mechanism $R_{ik}$ for rounds $t=\tau_{k-1}+1,\ldots,\tau_k$.}
\end{algorithm}

In \texttt{DRAM}, the entire horizon is divided into two phases: \textit{warm-start} phase and \textit{adaptive} phase.
As suggested by Theorem \ref{thm: cost of robustness}, our scheme fails when the ambiguity level is above a certain threshold $\tilde{\eta}$.
Therefore, the warm-start phase is designed to reduce ambiguity below that threshold.
In this phase, the principal would use ground truth $Y_t$ for verification.
Then, the principal moves to the adaptive phase, which is split into epochs.
At the beginning of each epoch, the principal uses the empirical distribution for estimation.
As the principal obtains more data, estimation is more accurate.
This allows it to design a mechanism with decreasing $\eta$ and therefore reduce additional cost for robustness.
The ambiguity parameter $\eta$ shrinks at a proper rate, so as to make sure truthfulness is preserved with high probability for each round.

During the entire algorithm, the principal uses agents' past report to estimate $p^*(x_j\mid x_i)$, which is the posterior distribution of reference agent's observation $x_j$ conditional on the focal agent $x_i$.
Note that the principal only has agents' reports $z$ but not the true observation $x$. 
This means truthfulness must be guaranteed at all times for estimation fidelity.
This adds another evidence on the necessity of truthfulness in addition to Theorem \ref{prop: truthfulness is necessary}.

\textbf{Inputs.}
Out of all the input variables, failure tolerance can be decided arbitrarily, and the rest depend on the agents.
The ambiguity threshold for all players is defined as $\tilde{\eta} = (1/2)\cdot \min_{i} \tilde{\eta_i}$, where the agent-specific threshold $\tilde{\eta}_i$ is computed according to Theorem \ref{thm: cost of robustness}:
\begin{align}
    \tilde{\eta}_i = \frac{1-\gamma_{(i)}}{2\| \mathbf{B}^{-1}_{(i)}\|(1+\gamma_{(i)})|\mathcal{Y}| + 2)}.
\end{align}

\textbf{Warm-start phase.}
The main objective of the warm-start phase is to reduce the principal's ambiguity below the threshold suggested by Theorem \ref{thm: cost of robustness}, so that distributionally robust mechanisms can be applied.
There are multiple approaches to reduce ambiguity, and here the principal learns $p^*(x_j \mid x_i) $ by collecting truthful reports from agents.
We incentivize truthfulness in this phase by using ground truth verification.
Suppose the principal can now obtain the ground truth $Y_t$ from an external expert.
With ground truth available, The principal can compare reports with it, and then reward according to a fact-checking mechanism $R_{it}(Z_{it}, Y_t)$.
This phase lasts $O(\log\log T)$ tasks, so cost is controlled even when ground truth is expensive.

\begin{lemma}[Fact checking under diagonal dominance]
    \label{lem: warm-start fact-checking}
    Assume that each true label $y$ appears with uniformly bounded probability $\underline{p} \leq p_Y(y) \leq \overline{p}$.
    If for all $y\in\mathcal{Y}$ and $x\neq y$, agent $i$'s skill $p_i$ satisfies the diagonal dominance property:
    \begin{equation}
        \blue{p_i(x \mid x)} \geq \frac{\overline{p}}{\underline{p}} \cdot p_i(x \mid y),
    \end{equation}
    \blue{then under the simple fact-checking rule $R_{it}(Z_{it}, Y_t) = s \cdot  \mathbf{1}\{Z_{it} = Y_t\}$ with sufficiently large scaling factor $s = s(p_{Y}, p_i)$, truthful reporting is a best strategy for agent $i$.}
\end{lemma}

The \textit{diagonal dominance} condition essentially assumes that agents are more likely to obtain the correct observation than to make a mistake.
Therefore, any lying behavior would decrease the probability of the report being correct, and fact-checking incentivizes truthfulness.
We note that it is generally impossible to design a fact-checking rule that guarantees truthfulness for arbitrary skill distribution $p_i$ and $p_Y$.
It is shown in \cite{lambert2011elicitation} that if a agent's observation has overlaps, i.e., the agent can have the same observation under two different labels, then there always exists an adversarial prior under which a fact-checking mechanism fails.

\textbf{Adaptive phase.}
After the ambiguity is lower than the threshold, the principal moves onto the adaptive phase.
This phase divides the entire time horizon into epochs, with each epoch double the size of the previous one.
In total, we would have $O(\log T)$ epochs.

At the beginning of each epoch, the principal would call two oracles: i) an offline estimation oracle for the reference distribution $p^*(x_j\mid x_i)$ (in Algorithm \ref{alg: dram} it is the empirical distribution estimator), and ii) an optimization oracle that computes the distributionally robust mechanisms by solving Eq.(\ref{eq: robust two-agent mechanism design}).
Then, the principal would use the same produced mechanism throughout the entire epoch, and no further computation is needed.
This indicates $\dram$ is \textit{computationally efficient} with $O(N\log T)$ total calls to both oracles.

At the same time, $\dram$ is also \textit{statistically efficient} for we have the following regret guarantee.

\begin{theorem}[Regret upper bound of \texttt{DRAM}] 
\label{thm: dram upper bound}
Consider the sequential mechanism design problem with $N$ agents and $T$ rounds.
\blue{Assume for each agent $i$, the fact-checking rule in the warm-start phase does elicit truthful reports, and the belief matrix $\mathbf{B}_{(i)}$ is invertible.}
\blue{Then there exists an event with probability at least $1-\varepsilon$, on which Algorithm \ref{alg: dram} achieves:}
\begin{itemize}
    \item truthfulness is guaranteed for all $N$ agents in all $T$ rounds;
    \item conditional on that event, expected total regret is at most
    \begin{equation}
        O\left(N\sqrt{T} \log(N\log T/\varepsilon)\right).
    \end{equation}
\end{itemize}

\end{theorem}

Theorem \ref{thm: dram upper bound} recovers the $O(\sqrt{T})$ terms typically seen in bandits and online learning literature \cite{lattimore2020bandit}.
In fact, oracle calls can be further reduced to $O(N\log\log T)$ when $T$ is known.
In \dram, we use the classical \textit{doubling trick} \cite{cesa2006prediction} from the online learning literature.
This trick does not require exact knowledge of the number of tasks.
(Although we need to know the magnitude $\log(T)$ to compute epochwise ambiguity parameter $\eta_k$s.)
The epoch schedule $\tau_k - \tau_{k-1} = T^{1-2^{-(k-1)}}\tau$ (similar to \cite{cesa2014regret}) maintains the same regret guarantee up to logarithmic terms, while requiring even fewer oracle calls $O(\log\log T)$.

\begin{corollary}[Regret upper bound with known $T$]
    \label{col: dram upper bound known t}
    \blue{Assume all conditions in Theorem \ref{thm: dram upper bound} holds.}
    Replace the epoch schedule in Algorithm \ref{alg: dram} with $\tau_k - \tau_{k-1} = T^{1-2^{-(k-1)}}\tau$, providing $O(\log\log T)$ epochs.
    Then there exists an event with probability at least $1-\varepsilon$, on which the principal simultaneously guarantees feasibility across all rounds for all agents, and the conditional expected total regret is at most
    \[
        O\left(N\sqrt{T} \log\left(\frac{N \log \log T}{\varepsilon}\right)\log\log T\right).
    \]
    
\end{corollary}

\subsection{Regret Lower Bound of Adaptive Mechanism Design}
We now show that \texttt{DRAM} is \emph{statistically optimal} up to logarithmic factors.
In particular, we prove a matching lower bound demonstrating that any policy that is feasible with high probability must incur expected regret of order at least $\Omega(N\sqrt{T})$.

\begin{theorem}
    \label{thm: lower bound}
    Consider the sequential mechanism design problem with $N$ agents and $T$ rounds.
    Fix any failure tolerance $\varepsilon \in (0, 1/4)$.
    For any (possibly randomized) non-anticipating reward policy \blue{that is feasible (satisfies constraints in Eq.\eqref{eq: two-agent mechanism design})} for all agents across all rounds with probability at least $1-\varepsilon$, there exists a type distribution $p_{\boldsymbol{X}}\in \Delta(\mathcal{Y^N})$ under which, the expected total regret \blue{under the feasibility event} is at least
    \begin{align*}
        \Omega(N\sqrt{T}).
    \end{align*}
\end{theorem}

The proof start by constructing a pair of statistically indistinguishable, two-agent problem instances whose corresponding cost-efficient feasible mechanisms are incompatible.
Specifically, any reward mechanism that is both feasible and near-optimal under one instance must either violate constraints or incur strictly larger payments under the other.
This incompatibility allows us to reduce adaptive mechanism design to a hypothesis testing problem, and we invoke Le~Cam’s two-point method to derive the lower bound.
\blue{To extend the result to $N$ agents, we apply a direct-sum construction. 
Agents are divided into pairs, and each pair independently selects one of the two-agent hard instances.
Under this construction, the principal is essentially handling $\lfloor N/2 \rfloor$ independent sequential games, thus the regret lower bound is linear in $N$.}

The lower bound together with its proof (see Appendix \ref{sec: defered proofs}) reveals that the regret bottleneck of adaptive mechanism design is the difficulty in learning players’ conditional beliefs, namely the posterior distributions $p^*(x_j \mid x_i)$ that govern incentives.
Because the lower bound is derived via a two-point argument, it does not explicitly depend on the alphabet size $d = |\mathcal{Y}|$.
However, since estimating a discrete distribution over $\mathcal{Y}$ incurs a minimax risk of order $\sqrt{d/T}$~\cite{han2015minimax}, we conjecture that the regret bound achieved by \texttt{DRAM} is also optimal in its dependence on $d$, up to logarithmic factors.

\subsection{Extension to General Estimator} 

An important observation of $\dram$ is that the estimation oracle and the optimization oracle are decoupled, connected only via the ambiguity parameter $\eta_k$.
This means that $\dram$ is flexible with estimators, so long as its estimation could satisfy the requirement in Eq.(\ref{eq: tv condition}).
Therefore, the empirical estimator could be swapped with any other distribution estimator that may better exploit and reflect the underlying structures of agents' skills.
Based on this, we propose the algorithm $\dram+$ working with general distribution estimators.

\begin{definition}[General Discrete Distribution Estimator]
\label{def: general estimator}
Let $q$ be a discrete distribution on space $\mathcal{Y}$.
Given $t$ samples independently and identically drawn from $q$, the generalized distribution estimator provides an estimation $\hat{q}_t$ such that with probability $1-\varepsilon$, we have
\begin{align}
    \tv (q,\hat{q}_t) \leq \eta_{\varepsilon}(t),
\end{align}
\end{definition}
\blue{whereas the function $\eta_\varepsilon(t)$ is non-increasing in $t$ and non-decreasing in $\varepsilon$.}

This estimation guarantees $\eta_\varepsilon(t)$ is commonly seen in the probably approximately correct (PAC) framework of statistical learning, where a better estimator should achieve a lower gap with higher probability and lower $t$.

Now we introduce $\dram+$ (Algorithm \ref{alg: dram+}), which modifies $\dram$ to work with general discrete distribution estimators in Definition \ref{def: general estimator}.
In $\dram+$, we don't restrict how epoch schedules are designed.
Generally, one should aim for a geometric epoch schedule, as this typically results in the best possible bounds and only $m = O(\log T)$ epochs.
Moreover, the ambiguity parameters now follow the guarantee $\eta_\varepsilon(t)$, in order to ensure truthfulness holds with high probability.

\begin{algorithm}[h]
\caption{Distributionally Robust Adaptive Mechanism+ (\texttt{DRAM+})}
\label{alg: dram+}
\DontPrintSemicolon
\SetKwInOut{Input}{Input}
\Input{ambiguity threshold $\tilde{\eta}$; failure tolerance $\varepsilon$; lower bound on observation frequency $0< \rho < \min_{i, x\in\mathcal{Y}}\mathbb{P}(X_i = x)$; distribution estimator $\mathcal{E}$.}


Compute the warm-start phase length as the smallest $\tau$ such that $\eta_{\varepsilon/Nm(d+1)}(\rho\tau/2) < \tilde{\eta}$.\;

For each agent $i$, assign a corresponding reference agent $j$.\;

\BlankLine
\textbf{\underline{Warm-start phase.}}\;
Follows the same procedure as in Algorithm \ref{alg: dram}.

\BlankLine
\textbf{\underline{Adaptive phase.}}\;

Define epoch schedule $\tau=\tau_0<\tau_1<\tau_2<\cdots < \tau_m = T$.\;

\For{$k=1,2,\ldots,m$}{
    Estimate reference distribution with the general distribution estimator for each $x_i\in\mathcal{Y} \cup \{\varnothing\}$:
    \[
    \hat{p}_{ik}(x_j \mid x_i) 
    \leftarrow \{Z_{it}, Z_{jt}\}_{t\le\tau_{k-1}}.
    \]
    
    Let ambiguity parameter 
    $\eta_k = \eta_{\varepsilon/Nm(d+1)}(\rho\tau_{k-1}/2)$.\;

    Compute the safety margin $\delta_{ik}$ and deploy the mechanism $R_{ik}$ the same way as in Algorithm \ref{alg: dram}.
}
\end{algorithm}

\begin{theorem}[Regret upper bound of $\dram +$]
\label{thm: dram-plus upper bound}
Consider the sequential mechanism design problem with $N$ agents and $T$ rounds.
\blue{Assume all conditions in Theorem \ref{thm: dram upper bound} holds.}
Then there exists an event with probability at least $1-\varepsilon - N(d+1) \cdot \sum_{k=1}^m \exp(-\rho \tau_{k-1}/8)$, on which Algorithm \ref{alg: dram+} achieves:
\begin{itemize}
    \item truthfulness is guaranteed for all $N$ agents in all $T$ rounds;
    \item conditional on that event, expected total regret of the algorithm is at most
    \begin{align}
        O\left(N\sum_{k=1}^m \eta_{\varepsilon/Nm(d+1)} (\rho \tau_{k-1}/2) \cdot (\tau_k - \tau_{k-1})\right).
    \end{align}
\end{itemize}
\end{theorem}

Compared to Theorem \ref{thm: dram-plus upper bound}, an additional overhead $N(d+1) \cdot \sum_{k=1}^m \exp(-\rho \tau_{k-1}/8)$ appears in the high probability guarantee.
This is due to a failed event when a player does not observe a certain label $x_i$ for enough number of times, and there are not enough data to recover $p^*(\cdot | x_i)$.
This failed event is universal, and without a closed-form description of the estimation gap, we cannot merge this probability with $\varepsilon$.
Nonetheless this term decays exponentially, and, with appropriately chosen schedule (such as $\tau_k - \tau_{k-1} = 2^{k-1}\tau$), it should be dominated by $\varepsilon$.


The central interpretation of $\dram+$ and Theorem \ref{thm: dram-plus upper bound} is that any estimation guarantees for discrete distribution can be immediately translated to mechanism regret guarantees.
This suggests that a principled reduction from online mechanism design to offline learning may indeed be possible, similar to the reduction from contextual bandits to offline estimation in \cite{simchi2022bypassing}.

Compared to the classical multi-armed bandit, the adaptive mechanism design problem is, in some sense, both simpler and harder.
On the one hand, the core challenge in bandit problems lies in the exploration–exploitation trade-off, since the arm-sampling policy in previous rounds affects observed data distributions. 
From that perspective, the mechanism design problem is simpler than bandits, since the underlying distributions remain unaffected by the principal’s mechanism decision as long as agents behave truthfully. 
On the other hand, when participants are rational, incentivizing truthfulness in an optimal way is nontrivial.
Agents’ incentives and skills are unknown, and any deviation can cause unpredictable dynamics. 
In contrast, there are no incentives involved in bandits, and each arm always gives truthful feedback.


\subsection{Discussions}

We collect some interesting observations from the Algorithm \ref{alg: dram} and \ref{alg: dram+} and their corresponding guarantees (Theorem \ref{thm: dram upper bound} and \ref{thm: dram-plus upper bound}).
These observations further demonstrate the generality of our results.


\textbf{Robustness to fluctuation/non-stationarity of agent performance.}
We assume each agent's skill (i.e., law $p_{i}(x_i\mid y)$) is consistent throughout the sequential tasks, an assumption not necessarily true in practice.
agents may under or over-perform in certain rounds compared to other rounds, resulting in skill fluctuation and non-stationarity.
In $\dram$, we apply a distributionally robust mechanism each round.
This robustness holds not only for estimation inaccuracy, but also to inaccuracy from other sources.
This means as long as the actual reference distribution stays within the ambiguity set defined by $\hat{p}_{ik}$ and $\eta_k$, agents are still incentivized to stay truthful.
Indeed, the principal could even widen up or narrow the ambiguity set by adjusting $\eta_k$, looking for more robustness or less cost.

\textbf{Robustness to adversary.}
For the same reason, the distributionally robust mechanism would also provide robustness to adversarial behavior from agents.
When an agent intentionally lies in a small portion of rounds, it would only slightly bias the estimation. 
As long as it does not surpass the ambiguity margin as designed in each epoch, the mechanism would not break down.
In addition, the assignment procedure of reference agents may provide additional defense.
An adversary would at most disrupt at most $2T$ out of $NT$ agent–task interactions in total (being one focal agent and one reference agent), possibly spread across different agents and tasks.

\textbf{Flexibility with reference agent assignment procedure.}
In $\dram$, each agent is assigned one corresponding reference agent $j$, to which her reports will be compared.
Any assignment procedure (deterministic or randomized) could be used for this process, and some could provide robustness to possible adverse agents.
As an example, suppose we use cyclic matching, where agent $i+1$ is assigned to agent $i$ as reference for $i<N$, and agent $1$ is assigned to agent $N$.
Under cyclic matching, any adversary would disrupt at most two agents, and the majority is not affected.
Furthermore, at the beginning of each epoch, we could rerun the procedure and assign new agents.
Such replacement generates little extra computational costs, since the principal needs to update estimation and regenerate mechanism anyway.

\textbf{Compatibility with delayed/batched feedback.}
In the practical setting, the feedback to the principal might not be immediately available and may come in batches \cite{chapelle2011empirical, mcmahan2013ad}. 
The delayed/batched feedback setting has been studied in multiple online learning and decision-making problems \cite{joulani2013online, gao2019batched}.
In $\dram$, the mechanisms are computed at the beginning of each epoch, and stay the same throughout.
This means $\dram$ naturally handles delayed and batched feedback since report data is only required for computation at the beginning of each epoch.
Particularly, Corollary \ref{col: dram upper bound known t} suggests that $O(\log\log T)$ epochs are already sufficient for the $O(\sqrt{T})$ bound up to logarithmic terms.
Nevertheless, such small epoch counts rely on a carefully designed epoch schedule.
For example, $\dram$ uses a geometric epoch schedule, under which it quickly adapts early on, and slows down when sufficient data are gained.
Deviation from the $O(\sqrt{T})$ bound may appear when the principal faces a different epoch schedule constraint \cite{perchet2016batched}.


\section{Experiments}

In this section, we perform numerical experiments to verify and demonstrate our proposed algorithm.

\textbf{Environment.}
We consider a sequential labeling game (as in Figure \ref{fig: toy_example} and Example \ref{eg: image labeling}) with $N=3$ agents and $d=3$ labels with a uniform prior $p_Y(y)=1/3$.
Each agent $i$ has a diagonally-dominant skill distribution $p_i(\cdot\mid y)$ that is symmetric across labels:
\[
p_i(x\mid y)=
\begin{cases}
\alpha_i, & x=y,\\[2pt]
\frac{1-\alpha_i}{d-1}, & x\neq y,
\end{cases}
\qquad
\alpha_i = 0.7 + \Big(i-\frac{N-1}{2}\Big)\cdot 0.02.
\]
Thus for $N=3$, $\alpha_0=0.68$, $\alpha_1=0.70$, and $\alpha_2=0.72$, with the remaining probability mass spread uniformly over the $d-1$ incorrect labels.
We use horizon $T=10^6$ and observation cost $c=0.3$.
During warm-start, the principal acquires the ground-truth label $Y_t$ from an external expert at cost $C_{\mathrm{lab}}=3.0$ per round.
We run $1000$ independent episodes.

\textbf{Algorithm setup.}
We implement the exact $\dram$ algorithm (Algorithm \ref{alg: dram}) in this simulation.
To match the theoretical parameterization, we compute
\[
\rho_{\text{true}}=\min_{i,\,x\in\mathcal{Y}} \mathbb{P}(X_i=x),
\qquad
\gamma_{(i)}=\max_{x\in\mathcal{Y}} \mathbb{P}(X_i=x),
\]
and the agent-wise robustness thresholds $\tilde{\eta}_i$ from Theorem~\ref{thm: cost of robustness}, then set 
\[
\tilde{\eta}_{\text{true}}=\min_i \tilde{\eta}_i,
\qquad
\tilde{\eta}_{\text{used}}=\min\!\left(0.9\,\tilde{\eta}_{\text{true}},\,1/\sqrt{2}\right),
\qquad
\rho_{\text{used}}=0.99\,\rho_{\text{true}}.
\]
Given $\varepsilon=10^{-3}$, we plug $(\tilde{\eta}_{\text{used}},\rho_{\text{used}})$ into the warm-start length formula in Algorithm~\ref{alg: dram} to obtain $\tau$. For this setting, $\tau$ is on the order of $10^5$, so the warm-start phase occupies only a moderate fraction of the horizon. In the warm-start phase, we use the simple fact-checking mechanism: reward both agent $1$ if their report agrees, and $0$ if not.

\begin{figure}
    \centering
    \includegraphics[width=0.98\linewidth]{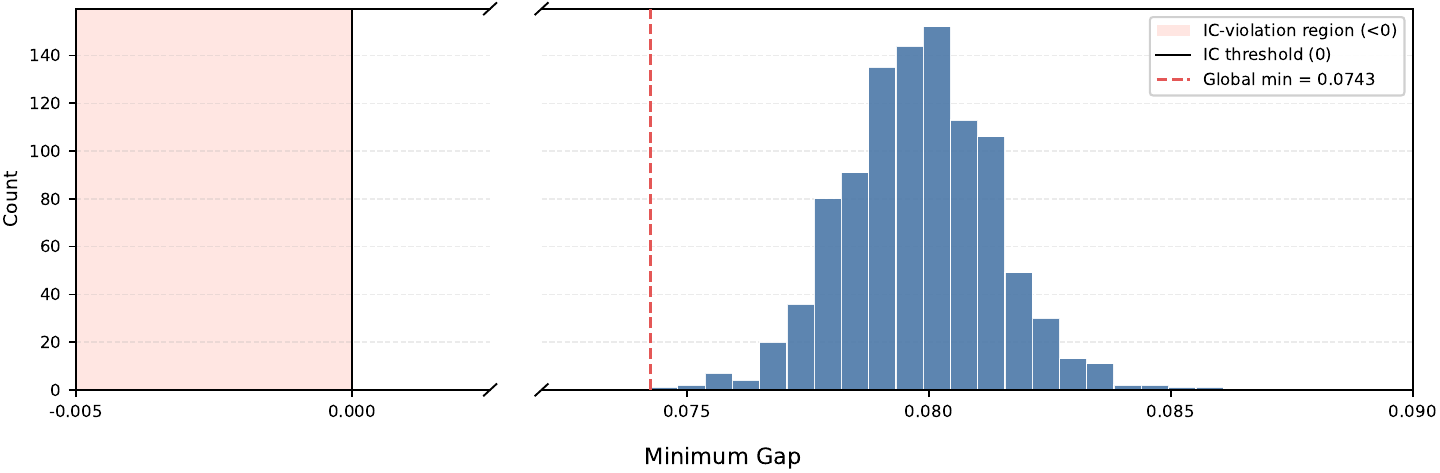}
    \caption{Minimum reward gap between truthful reporting and other pure strategies across 1000 runs of a sequential labeling game. Negative gap means the constraints are violated. In this simulation, the minimum gap distribution is well separated from $0$, meaning that truthful reporting dominates other strategies by a considerable margin, and $\dram$ guarantees truthfulness even with spare robustness.}
    \label{fig: ic gap}
\end{figure}

\textbf{Truthfulness checks.}
We verify truthfulness via a retrospective approach.
We set all participating agents to be always truthful.
At the beginning of every epoch, we perform a truthfulness check using the true joint distribution $p^*(X_i, X_j)$.


We compute the truthful expected utility
\[
U_{ik}^{\mathrm{truth}}
=
\mathbb{E}_{(X_i,X_j)\sim p^*}\!\left[R_{ik}(X_i,X_j)\right]-c,
\]
and compare it against two families of deviations:
\begin{itemize}
    \item \textbf{Lazy strategies:} the best constant report (without observation) $z\in\mathcal{Y}$ with no observation cost,
    \[
    U_{ik}^{\mathrm{lazy}}
    =
    \max_{z\in\mathcal{Y}}
    \mathbb{E}_{X_j\sim p^*}\!\left[R_{ik}(z,X_j)\right].
    \]
    \item \textbf{Misreporting strategies:} all deterministic mappings from observation to $g:\mathcal{Y}\to\mathcal{Y}$ excluding the identity,
    \[
    U_{ik}^{g}
    =
    \mathbb{E}_{(X_i,X_j)\sim p^*}\!\left[R_{ik}(g(X_i),X_j)\right]-c.
    \]
\end{itemize}


We define the IC gap for agent $i$ in epoch $k$ as
\[
\mathrm{Gap}_{ik}
=
U_{ik}^{\mathrm{truth}}
-
\max\!\Big\{U_{ik}^{\mathrm{lazy}},\ \max_{g\neq \mathrm{id}} U_{ik}^{g}\Big\},
\]
and, for each episode, record the minimum $\mathrm{Gap}_{ik}$ across all agents and epochs.
If any $\mathrm{Gap}_{ik}$ is negative, we count the episode as an IC violation.
Figure~\ref{fig: ic gap} shows the histogram of per-episode minimum IC gaps.


\begin{figure}
    \centering
    \includegraphics[width=0.98\linewidth]{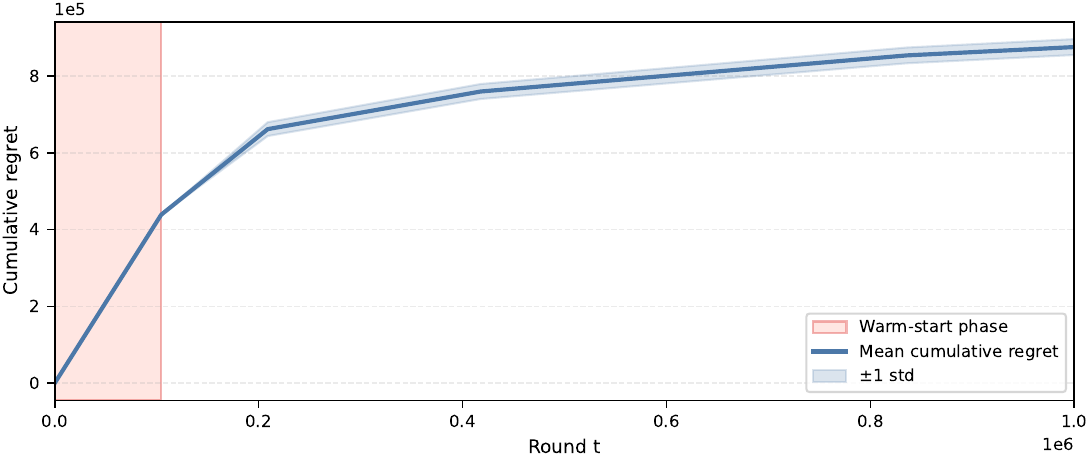}
    \caption{Average cumulative regret over time across 1000 runs of a sequential labeling game. The first $\sim10^5$ rounds are warm-start phase, and then comes with doubling epoch lengths. The regret curve is piecewise linear, as the expected round-wise regret within each epoch stays unchanged. The geometric epoch schedule ensures $O(\sqrt{T})$ regret.}
    \label{fig: regret}
\end{figure}

\textbf{Regret checks.}
We collect the cumulative regret at each round within each episode of the game.
In addition to the regret notion defined in Section~\ref{sec: adaptive mech design}, we include the warm-start verification cost in the regret formula.
Specifically, we plot the cumulative regret up to time $t$:
\[
\mathrm{Reg}(t)
=
\sum_{s=1}^{t}\left(\sum_{i=1}^{N} R_{is} - Nc + \mathbf{1}\{s\le \tau\}\,C_{\mathrm{lab}}\right),
\]
and report the mean and standard deviation across $1000$ episodes.
Figure~\ref{fig: regret} shows the resulting regret trajectory.


\textbf{Results.}
$\dram$ passes the truthfulness checks as across the $1000$ episodes, we observe \emph{no} truthfulness violations.
The global minimum gap is approximately $0.0743 > 0$, and the distribution of per-episode minimum gaps is well separated from zero.
This indicates that, in a setting that exactly matches our assumptions and with theoretically chosen parameters $(\tau, \delta)$, $\dram$ indeed implements a truthful mechanism in practice.

$\dram$ also consistently achieves the $O(\sqrt{T})$ regret, as shown in Figure \ref{fig: regret}.
In this simulation we have 5 epochs in total.
Within each epoch, the mechanism stays unchanged, therefore the cumulative regret curve is piecewise linear.
In summary, this experiment demonstrates the efficiency and robustness of the vanilla $\dram$ algorithm, and validates the correctness of Theorem \ref{thm: dram upper bound}.
In fact, the existence of extra IC gap seems to suggest that further refinement are possible, as the current parameters are set for theoretical proofs rather than optimized for practical implementations.

\section{Conclusions}
In this paper, we designed an adaptive mechanism for the sequential mechanism design problem.
The studied problem assumes rational feedback compared to the prediction with expert advice problem from online learning, and relaxes the common knowledge assumption compared to the peer prediction problem from mechanism design.
Drawing insights from both fields, our proposed mechanism ensures truthful behaviors with high probability, while achieving optimal payment regret.
It also remains robust and adaptable in changing environments.

Looking forward, our work motivates interesting questions.
In Section \ref{sec: robust mechanism design}, the mechanism design problem is formulated as a linear optimization problem, with truthfulness encoded as constraints.
A key idea of our algorithm is to solve a distributionally robust variant of this problem while gradually learning the relevant constraints over time.
This principle seems to be broadly applicable: since many decision-making problems can be cast as optimization tasks, the same approach might extend naturally to online, adaptive, or sequential variants of other real-world decision-making problems beyond mechanism design.

\section*{Acknowledgments}
The authors thank Rui Ai, Jiachun Li, Chonghuan Wang, Yunzong Xu, Yuan Zhou, Feng Zhu for helpful comments and suggestions.
The authors are especially grateful to Runhuan Wang for his careful reading of the manuscript and detailed technical feedback, which significantly improved the paper.






\bibliographystyle{ACM-Reference-Format}
\bibliography{refs}

\newpage
\appendix
\begin{center}
    \large \textbf{Appendix}
\end{center}

\section{Deferred Proofs}
\label{sec: defered proofs}

\subsection{Proof of Proposition \ref{prop: truthfulness is necessary}}

Suppose the past rounds' report history is $\mathbf{z}_1, \mathbf{z}_2, \dotsb, \mathbf{z}_{t-1}$.
All probability laws and strategies discussed below are conditional on such history.

\blue{We first show that the optimal value $V(\boldsymbol{Z})$ is obtained by the Bayes optimal decision rule $\delta^*$ that maximizes posterior expected objectives:
$$
    \delta^*(\boldsymbol{z}) = \max_{a\in\mathcal{A}} \sum_{y\in\mathcal{Y}} u(y, a)\mathbb{P}(Y_t = y \mid \boldsymbol{Z}=\boldsymbol{z}), \quad \forall \boldsymbol{z}.
$$
This is because we have the tower property:
$$
    V(\boldsymbol{Z}) = \max_\delta \mathbb{E}[u(Y_t, \delta(\boldsymbol{Z}))] = \max_\delta \mathbb{E}[\mathbb{E}[u(Y_t, \delta(\boldsymbol{Z}))\mid \boldsymbol{Z}]]
$$
Hence for each report realization $\boldsymbol{Z} = \boldsymbol{z}$, the action only affects the inner conditional expected objective itself, and the optimal decision rule is achieved by pointwise picking the optimal action for each $\boldsymbol{z}$.}

\blue{Now we introduce some definitions to connect our proposition with Blackwell's informativeness theorem.}
In round $t$, denote agent $i$'s \textit{report strategy} by $s_{it}: \mathcal{Y} \to \Delta(\mathcal{Y})$, where $s_{it}(z\mid x)$ is the probability to report $z$ given $x$.
Define agent $i$'s \textit{information structure} $\sigma_{it}$ as her conditional probability law to report $z$ conditional on true label $y$. Then we have
$$
    \sigma_{it} (z \mid y) = \sum_{x\in \mathcal{Y}} s_{it}(z\mid x) p_i(x\mid y).
$$

Alternatively, we can write $\sigma_{it} = s_{it}\circ p_i$ as the information structure is induced by strategy $s_{it}$.
We let $\boldsymbol{\sigma} = (\sigma_{1 t}, \dotsb, \sigma_{Nt})$ be the information structure profile of all agents, \blue{and let $\boldsymbol{s} = (s_{1 t}, \dotsb, s_{Nt})$ be the strategy profile of all agents.
With a slight abuse of notation, for a given decision-making problem $(\mathcal{A}, u)$, we define the Bayes-optimal value under information structure profile $\boldsymbol{\sigma}$ and prior belief on true label $p_Y$ as:
$$
    V(\boldsymbol{\sigma}) = \sum_{\boldsymbol{z}} \mathbb{P}(\boldsymbol{Z} = \boldsymbol{z})\left(\max_{a\in\mathcal{A}} \sum_{y\in\mathcal{Y}}u(y,a)\frac{\boldsymbol{\sigma}(\boldsymbol{z} | y)p_Y(y)}{p_{\boldsymbol{Z}}(\boldsymbol{z})}\right)
$$ 
Notice that a report profile $\boldsymbol{Z}$ is generated by applying the information structure $\boldsymbol{\sigma}$ to the underlying ground truth $Y_t$.
Also, in the beginning we established that the optimal value $V(\boldsymbol{Z})$ can be obtained by adopting the Bayes-optimal decision rule.
Therefore, for a given information structure profile $\boldsymbol{\sigma}$ and its corresponding report $\boldsymbol{Z}$, the definition of $V(\boldsymbol{\sigma})$ and $V(\boldsymbol{Z})$ are equivalent, and we can use them interchangeably.
}



We now state the Blackwell's informativeness theorem.
The theorem states three equivalent conditions, of which we only use the following two.

\begin{lemma}[\cite{blackwell1953equivalent}]
\label{lem: blackwell}
Suppose we have two information structure profiles $\boldsymbol{\sigma}$ and $\boldsymbol{\sigma}'$, the following two conditions are equivalent:
\begin{itemize}
    \item \blue{$V(\boldsymbol{\sigma}) \geq V(\boldsymbol{\sigma}')$ for every prior belief $p_Y$ and every decision problem $(\mathcal{A}, u)$.}
    \item There exists a stochastic map $\Gamma$ such that $\boldsymbol{\sigma}^{\prime}=\Gamma \circ \boldsymbol{\sigma}$. That is, $\boldsymbol{\sigma}^{\prime}$ is a garbling of $\boldsymbol{\sigma}$.
\end{itemize}
\end{lemma}
Note that this lemma applies for general $\boldsymbol{\sigma}$, even those where agent $i$'s report may depend on other agents' observation.
In our setting, we are only concerned with the subset of structures where each agent's report is independent from others' observation.

\textbf{Sufficiency.} 
We first prove that permutative strategies are sufficient for maximal quality.
Fix any agent $i$, and let $\pi$ be her permutative strategy.
Any report strategy $s$ is a garbling of $\pi$, since $s = (s \circ \pi^{-1}) \circ \pi$.
Therefore the corresponding information structure $\sigma_i' = s\circ p_i$ is also a garbling of $\sigma_i = \pi \circ p_i$.
This applies for all $i$, therefore any information structure profile $\boldsymbol{\sigma_{s}}$ induced by strategy profile $\boldsymbol{s}$ is also a garbling for $\boldsymbol{\sigma_{\pi}}$ induced by any permutation strategy profile $\boldsymbol{\pi}$.
\blue{From Lemma \ref{lem: blackwell}, we therefore have $V(\boldsymbol{\sigma}_{\pi}) \geq V(\boldsymbol{\sigma}_{s})$, and equivalently $V(\boldsymbol{\boldsymbol{Z}}_{\pi}) \geq V(\boldsymbol{Z}_{s})$ for every $p_Y$ and decision problem under a fixed skill profile $p_{\boldsymbol{X}|Y}$.
Finally, notice that for every skill profile the above inequality always hold, so by taking union over all possible skill profile, we can recover the original maximal quality statement in Theorem \ref{prop: truthfulness is necessary}.}

\textbf{Necessity.} 
\blue{
We now show that permutative strategies are necessary.
Since the strategy profile produces a report that dominates other reports for any joint law $p_{Y,\boldsymbol{X}}$ and decision problems, we can fix a specific naive skill profile, and the information structure profile generated by such strategy profile  would dominate any other information structure under this specific naive skills.
}

\blue{
Fix any focal player $i$.
Consider a skill profile $p_{\boldsymbol{X} \mid Y }$ whereas only focal player $i$ is $100\%$ accurate: $p_i(x_i\mid y) = \mathbf{1}\{x_i = y\}$, while the rest players always observes the same symbol arrangements $x^* \in \mathcal{Y}^{N-1}$: $p_{-i} (x_{-i} \mid y) = \mathbf{1} \{x_{-i} = x^*\}$.
}

\blue{
Combined with this skill profile, the report strategy profile produces an information profile $\sigma$ that dominates in value over other strategy-induced information profile $\sigma'$ under any $p_Y$.
From Lemma \ref{lem: blackwell} this implies that $\sigma' = \Gamma \circ \sigma$.
Consider the information profile $\sigma'$ produced by the following strategy: player $i$ stays always truthful, while others uses the same report strategy as $\sigma$.
In other words, this gives us the following:
\begin{align*}
    \sigma((z_i, z_{-i}) \mid y) &= s_i (z_i \mid y) \cdot s_{-i}(z_{-i} \mid x^*) \\
    \sigma'((z_i, z_{-i}) \mid y) &= \mathbf{1}\{z_i = y\} \cdot s_{-i}(z_{-i} \mid x^*)
\end{align*}
and $\sigma'$ is a garbling of $\sigma$.
}

\blue{
We first show the truthful strategy $\mathbf{1}\{z_i = y\}$ must be a garbling of player $i$'s current strategy $s_i(z_i \mid y)$.
We have
\begin{align*}
    \mathbf{1}\{z_i = y\} \cdot s_{-i}(z_{-i} \mid x^*) = \sum_{w} \Gamma(z_i, z_{-i} \mid w_i, w_{-i}) \cdot s_i (w_i \mid y) \cdot s_{-i}(w_{-i} \mid x^*)
\end{align*}
Marginalizing out $z_{-i}$ on the left hand side gives that
\begin{align*}
    \sum_{z_{-i}} \mathbf{1}\{z_i = y\} \cdot s_{-i}(z_{-i} \mid x^*) &= \mathbf{1}\{z_i = y\},
\end{align*}
and on the right
$$
    \sum_{w} \left( \sum_{z_{-i}} \Gamma(z_i, z_{-i} \mid w_i, w_{-i}) \right) \cdot s_i (w_i \mid y) \cdot s_{-i}(w_{-i} \mid x^*)
$$
which is a stochastic kernel from $\left(w_i, w_{-i}\right)$ to $z_i$. Then
$$
\mathbf{1}\left\{z_i=y\right\}=\sum_{w_i, w_{-i}} \widetilde{\Gamma}\left(z_i \mid w_i, w_{-i}\right) S_i\left(w_i \mid y\right) S_{-i}\left(w_{-i} \mid x^*\right)
$$
Now average over $w_{-i}$ according to $S_{-i}\left(\cdot \mid x^*\right)$, and define
$$
\bar{\Gamma}\left(z_i \mid w_i\right):=\sum_{w_{-i}} \widetilde{\Gamma}\left(z_i \mid w_i, w_{-i}\right) s_{-i}\left(w_{-i} \mid x^*\right) .
$$
Then $\bar{\Gamma}$ is a stochastic kernel from $w_i$ to $z_i$, and we obtain that for all $y \in \mathcal{Y}$
$$
\mathbf{1}\left\{z_i=y\right\}=\sum_{w_i} \bar{\Gamma}\left(z_i \mid w_i\right) s_i\left(w_i \mid y\right),
$$
}

\blue{
We now show that such property $\mathbf{1}\{z_i = x_i \} = \bar{\Gamma} \circ s$ indicates that $s$ must be a permutative strategy $\pi$.}

First, we prove that such a strategy cannot have overlapped labels.
Consider the truthful strategy $s'(z\mid x) = \mathbf{1}\{z=x\}$, we would have
$$
1 = s'(x\mid x) = \sum_{z\in \mathcal{Y}} \Gamma(x \mid z)s(z\mid x).
$$
Since $\sum_{z\in \mathcal{Y}} s(z\mid x) = 1$, for each $z$ where $s(z\mid x) > 0$, we must have $\Gamma(x\mid z) = 1$, otherwise the sum of weighted average would fall short of $1$.
Now, assume we have such overlapped labels $z$ and the corresponding $x_1$ and $x_2$, it would mean that $\Gamma(x_1 \mid z) = \Gamma(x_2 \mid z) = 1$.
But $\Gamma(\cdot \mid z)$ is a single probability distribution, so it cannot assign probability $1$ to two different outcomes simultaneously, forming a contradiction and we prove the overlapped signals cannot exist.

Since $s$ cannot have overlapped labels, by counting we know each observation must be mapped to one and only one label, meaning it is a permutative strategy.
Finally, performing this arguments over all player $i$ proves the necessity.

\textbf{Suboptimality of laziness.} Finally, we show that the lazy option (directly report according to a prior belief $\hat{p}_Y$) is strictly worse than observation with permutation strategy regardless of the lazy reporting strategy $\hat{p}_Y$ used.
Consider the case with two labels, where $p_Y(y) = 1/2$ for both $y \in \mathcal{Y}$ and $p_i = \mathbf{1}\{x = y\}$.
The induced information structure is $\sigma_Y(z \mid y) = \hat{p}_Y(z)$, and the information structure from a permutation strategy is $\sigma_{\pi} = \pi \circ p_i$.
Actually $\Gamma = \sigma_Y \circ \pi^{-1}$ makes $\sigma_Y$ a garbling of $\sigma_\pi$, since $\sigma_Y \circ \pi^{-1} \circ \pi \circ p_i = \sigma_Y \circ p_i = \sigma_Y$.
However, $\sigma_\pi$ is not a garbling of $\sigma_Y$ since the corresponding row-stochastic matrix of $\sigma_Y$ is rank $1$ but $\sigma_\pi$'s row-stochastic matrix is rank $2$.
Therefore Lemma \ref{lem: blackwell} suggests lazy option is strictly dominated by observation with permutation strategy.

$\square$

\subsection{Proof of Theorem \ref{thm: optimal cost of 2-agent peer prediction}}
\label{appendix: proof of optimal cost of 2-agent peer prediction}
\textbf{Feasibility.}
Suppose $\mathbf{B}$ is invertible. Then for arbitrary matrix $\mathbf{M}$, there exists $\mathbf{R} = (\mathbf{B}^{-1} \mathbf{M})^\intercal$ such that $\mathbf{B}\mathbf{R}^\intercal = \mathbf{M}$.
Notice that the entry $\mathbf{M}_{xy}$ is exactly agent $1$'s expected reward given she observes label $x$ and reports label $y$.
Hence for our purposes, we can construct an $\mathbf{M}$ whose diagonal entries are greater than $c$, and off-diagonal entries are less than $c$, then the corresponding $\mathbf{R}$ satisfies the first two constraints.
(Actually, if $\mathbf{B}$ is invertible, the first two constraints of \eqref{eq: LP formation for two-agent mechanism design} can be satisfied for arbitrary right-hand side values.)

Now we consider the third constraint $\mathbf{R}\mathbf{d} \leq \mathbf{0}$.
Notice that we have $\mathbf{d}^\intercal = \sum_x \mathbb{P}(X_i = x) \cdot \mathbf{B}_{x:}$.
Therefore letting $\mathbf{B}\mathbf{R}^\intercal = \mathbf{M}$, we have 
\begin{align*}
    \mathbf{R}\mathbf{d} &= \sum_x \mathbb{P}(X_i = x)\cdot \mathbf{R}  (\mathbf{B}_{x:})^\intercal \\
    &= \sum_x \mathbb{P}(X_i = x)\cdot (\mathbf{M}_{x:})^\intercal.
\end{align*}
Therefore, to satisfy all the three constraints, we need to find a matrix $\mathbf{M}$ whose linear combinations of its rows under coefficients $\{\mathbb{P}(X_i = x)\}_{x\in \mathcal{Y}}$ yield a vector with non-positive entries.
Let $\gamma = \max_x \mathbb{P}(X_i = x)$, we know that $\gamma <1$.
Then if we let all diagonal values of $\mathbf{M}$ be $c$, and all off-diagonal values of $\mathbf{M}$ equals $-c\gamma/(1-\gamma)$, then we would have for all $x'\in \mathcal{Y}$, 
\begin{align*}
    (\mathbf{R}\mathbf{d})_{x'} &= \sum_x \mathbb{P}(X_i = x)\cdot (\mathbf{M}_{xx'})^\intercal \\
    &= \mathbb{P}(X_i = x')\cdot c + (1-\mathbb{P}(X_i=x')) \cdot (-c\gamma/(1-\gamma)) \\
    &\leq \gamma \cdot c + (1-\gamma) \cdot (-c\gamma/(1-\gamma)) \\
    &= 0.
\end{align*}
Hence such a matrix $\mathbf{M}$ exists and the corresponding $\mathbf{R}$ is a feasible solution.

\textbf{Optimality.}
Notice that the objective is essentially $\sum_{x} \mathbb{P}(X_i = x) \mathbf{M}_{xx}$. Since we constructed $\mathbf{M}$ with all diagonal values being $c$, the objective value is $c$.
The first constraint is binding.
Smaller objective is not possible as it would require $\mathbf{M}_{xx} < c$ for some $x$, violating the first constraint.

$\square$

\subsection{Proof of Theorem \ref{thm: robustness}}
It is more convenient to use the matrix notation (see \eqref{eq: LP formation for two-agent mechanism design}).
The condition is essentially saying $\|\mathbf{B} - \mathbf{B}^*\|_{\infty} \leq \delta/\kappa$ and $\|\mathbf{d} - \mathbf{d}^*\|_1 \leq \delta/\kappa$, where $\|\cdot \|_{\infty}$ is the matrix norm induced by vector $\infty$-norm.
(It is essentially the maximum absolute row sum of the matrix.)

Therefore, we have
\begin{align*}
    \max_{x,y} |\mathbf{B}\mathbf{R}^\intercal - \mathbf{B}^*\mathbf{R}^\intercal|_{xy} & = \max_{x,y} \sum_z (\mathbf{B}-\mathbf{B}^*)_{xz}\mathbf{R}_{yz} \\
    &\leq \max_x \|(\mathbf{B}-\mathbf{B}^*)_{x:}\|_1 \cdot \kappa \\
    &\leq 2 \cdot (\delta/2\kappa) \cdot \kappa\\
    & = \delta.
\end{align*}
Here it is crucial to notice that $\|(\mathbf{B}-\mathbf{B}^*)_{x:}\|_1 = 2\tv\left(p(\cdot \mid x_i),p^*(\cdot \mid x_i)\right)$.
Following the same procedure, we can show that $\|\mathbf{R}(\mathbf{d} -\mathbf{d}^*)\|_\infty \leq \delta$.

Therefore, the constraints in \eqref{eq: robust two-agent mechanism design} shift by at most $\delta$, which means the $\delta$-margin mechanism $R_i$ satisfies \eqref{eq: two-agent mechanism design}.

$\square$

\subsection{Proof of Theorem \ref{thm: bounds on payments of robust mechanism}}

\textbf{Worst-case payment.} 
We still use the matrix formulation for the problem (see \eqref{eq: LP formation for two-agent mechanism design}).
Under this notation, the problem becomes
\begin{align*}
    \min_{\mathbf{R}} \quad& \kappa \\ 
    \textrm{s.t.} 
    \quad& \|\mathbf{R}\|_{\max} \leq \kappa, \\
    \quad& (\mathbf{BR}^\intercal)_{xx} \geq c+\delta, \quad \forall x \in \mathcal{Y}\\
    \quad& (\mathbf{BR}^\intercal)_{xy} \leq c-\delta, \quad \forall x \neq y \in \mathcal{Y}\\
    \quad& \mathbf{Rd}\leq -\delta \cdot \mathbf{1}.
\end{align*}

We call this problem $\lp(p, c, \delta)$, since it is a linear programming problem with distribution $p$, cost $c$ and margin $\delta$. 
Notice that if $(\kappa, \mathbf{R})$ is a feasible solution to $\lp(p, c, 0)$, and $(\kappa', \mathbf{R}')$ is a feasible solution to $\lp(p, 0, 1)$, then $(\kappa+\delta\kappa', \mathbf{R} + \delta \mathbf{R}')$ is a feasible solution to $\lp(p, c, \delta)$. Therefore, we can construct upper bounds of $\lp(p, c, \delta)$ by constructing upper bounds on $\lp(p, c, 0)$ and $\lp(p, 0, 1)$ separately.
We apply the same strategy as proof of Theorem \ref{thm: optimal cost of 2-agent peer prediction}, that is, to consider the intermediate solution $\mathbf{M} = \mathbf{B}\mathbf{R}^\intercal$.
The mechanism can be easily acquired by $\mathbf{R} = (\mathbf{B}^{-1}\mathbf{M})^\intercal$.
With this reformulation (see Section \ref{appendix: proof of optimal cost of 2-agent peer prediction} for details), the problem can be constructed as 
\begin{align*}
    \min_{\mathbf{M}} \quad& \kappa \\ 
    \textrm{s.t.} 
    \quad& \|\mathbf{B}^{-1}\mathbf{M}\|_{\max} \leq \kappa, \\
    \quad& \mathbf{M}_{xx} \geq c+\delta, \quad \forall x \in \mathcal{Y}\\
    \quad& \mathbf{M}_{xy} \leq c-\delta, \quad \forall x \neq y \in \mathcal{Y}\\
    \quad& \mathbf{M}^\intercal \mathbf{d}' \leq -\delta \cdot \mathbf{1},
\end{align*}
where $\mathbf{d}'_x = \mathbb{P}(X_i = x)$.

\blue{First, consider a feasible solution $(\kappa, R_i)$. 
When $|R_i| < \kappa$, we then have for each $x \in \mathcal{Y}$
$$
    \mathbf{M}_{xx} = \sum_{x'\in \mathcal{Y}} \mathbf{B}_{xx'} \mathbf{R}_{xx'} \leq \sum_{x'\in \mathcal{Y}} \mathbf{B}_{xx'} \kappa = \kappa.
$$
}
The lower bound is now apparent, since when absolute maximal payment goes under $c+\delta$ we violate the constraint $\mathbf{M}_{xx} \geq c+\delta$.

\textit{Upper bounds of $\lp(p, c, 0)$.} 
Following the same construction as Appendix \ref{appendix: proof of optimal cost of 2-agent peer prediction}, we let $\mathbf{M}$ has all diagonal values being $c$, and all off-diagonals equal $-c\gamma/(1-\gamma)$, where $\gamma = \max_x \mathbb{P}(X_i = x)$. 
Then $\mathbf{M}$ satisfies the three constraints on matrix.
With this $\mathbf{M}$, we have
\begin{align*}
    \|\mathbf{B}^{-1}\mathbf{M}\|_{\max}   \leq  \|\mathbf{B}^{-1}\mathbf{M}\|_2   &\leq   \|\mathbf{B}^{-1}\|_2 \|\mathbf{M}\|_2  \\ 
    &= \|\mathbf{B}^{-1}\|_2  \cdot \frac{c}{1-\gamma}\max(1, \mid1-\gamma\cdot |\mathcal{Y}| \mid)\\
    &\leq \|\mathbf{B}^{-1}\|_2  \cdot \frac{c(\gamma|\mathcal{Y}| + 1)}{1-\gamma}.
\end{align*}
Here, all eigenvalues of $\mathbf{M}$ can be easily calculated since $\mathbf{M}$ is a combination of identity matrix $\mathbf{I}$ and all‑ones matrix $\mathbf{J}$, whose eigenvalues are known.
In the end we can take $\kappa \leq \|\mathbf{B}^{-1}\|_2 \cdot c(\gamma|\mathcal{Y}| + 1)/(1-\gamma)$.

\textit{Upper bounds of $\lp(p, 0, 1)$.}
Similarly, construct $\mathbf{M}'$ with diagonal $1$ and off-diagonal $-(1+\gamma)/(1-\gamma)$.
This $\mathbf{M}’$ satisfies all three constraints on matrix.
A similar argument gives us
\begin{align*}
    \kappa' \leq \|\mathbf{B}^{-1}\|_2 \cdot \frac{(1+\gamma)|\mathcal{Y}|+2}{1-\gamma}.
\end{align*}

Combining the two upper bounds, it means that $(\kappa+\delta\kappa', \mathbf{M}+\delta\mathbf{M}')$ is a feasible solution, and we end up with an upper bound on $\lp(p, c, \delta)$, which is:
\begin{align*}
    \kappa \leq \|\mathbf{B}^{-1}\|_2 \cdot \left(c\cdot\frac{\gamma|\mathcal{Y}| + 1}{1-\gamma} + \delta \cdot\frac{(1+\gamma)|\mathcal{Y}|+2}{1-\gamma}\right).
\end{align*}

\textbf{Expected payment.} 
The solution $(\kappa+\delta\kappa', \mathbf{M}+\delta\mathbf{M}')$ ensures that the constraint $\mathbf{M}_{xx} \geq c+\delta$ is binding.
Therefore, the expected payment under truthful equilibrium is $\mathbb{E}_p[R_i (X_i, X_j)] = \sum_x \mathbb{P}(X_i = x)\mathbb{E}_p[R_i(X_i, X_j)\mid X_i] = c+\delta$.

$\square$

\subsection{Proof of Theorem \ref{thm: cost of robustness}}
\blue{For each $\delta > 0$, by the second part of Theorem \ref{thm: bounds on payments of robust mechanism}, there exists a feasible solution $(\kappa(\delta), R_i)$ of Eq.\eqref{eq: robust two-agent mechanism design} whose $\kappa(\delta)$ satisfies Eq.\eqref{eq: kappa upper bound} and  expected payment is $c+\delta$.
To ensure agent $i$ stays truthful, according to Theorem \ref{thm: robustness}, we need to find a margin $\delta$ such that $\delta/2\kappa(\delta) \geq \eta$.
From Corollary \ref{col: robustness-floor}, it suffices to find a $\delta$ such that
\begin{align*}
    \frac{\delta}{2\kappa(\delta)} \geq \frac{\delta \cdot (1-\gamma)}{2\|\mathbf{B}^{-1}\|_2 (c\cdot (\gamma|\mathcal{Y}| + 1) + \delta \cdot ((1+\gamma)|\mathcal{Y}|+2))} = \eta,
\end{align*}
Solving the second inequality equation gives us
\begin{align*}
    \delta = \frac{2\|\mathbf{B}^{-1}\|_2 (\gamma|\mathcal{Y}| + 1)\cdot\eta}{(1-\gamma) - 2\|\mathbf{B}^{-1}\|_2((1+\gamma)|\mathcal{Y}| +2)\cdot \eta}\cdot c
\end{align*}
And under this margin we have a robust mechanism that guarantees truthfulness.}

If the actual distribution $p^*$ is in the required ambiguity set, it shifts this expected payment by at most an additional $\delta$, making the final expected payment at most $c+2\delta$.
Substituting it with the above selected margin $\delta$ gives the final result.

$\square$

\subsection{Proof of Lemma \ref{lem: warm-start fact-checking}}
Suppose agent $i$'s observation $X_{it} = x$.
Then from Bayes' rule we have
\begin{align*}
    \mathbb{P}(Y_t = x \mid X_{it} = x) &= p_i(x\mid x)\; p_Y(x)\;/\;\mathbb{P}(X_{it} = x) \geq \underline{p} \cdot p_i(x\mid x)\;/\;\mathbb{P}(X_{it} = x)\\
    \mathbb{P}(Y_t = y \mid X_{it} = x) &= p_i(x\mid y)\; p_Y(y)\;/\;\mathbb{P}(X_{it} = x) \leq \overline{p} \cdot p_i(x\mid y)\;/\;\mathbb{P}(X_{it} = x)
\end{align*}
for any $y\in\mathcal{Y}$.
The diagonal dominance then implies that $\mathbb{P}(Y_t = x \mid X_{it} = x) \geq \mathbb{P}(Y_t = y \mid X_{it} = x)$ for any $y\in\mathcal{Y}$. Hence a truthful strategy $Z_{it} = X_{it}$ maximizes reward under $s \mathbf{1}\{Z_{it} = Y_t\}$ for any $s$. 

\blue{
Finally, the best lazy strategy that reports $z = \operatorname{argmax}_{y} p_Y(y)$, which generates reward $(s  \max p_Y(y)$.
Observing and then reporting yields expected reward $s\sum_{y} p_Y(y) p_i(y\mid y) - c$.
Therefore to ensure truthfulness, we need that
$$
    s\left(\sum_{y} p_Y(y) p_i(y\mid y) - \max p_Y(y) \right)  \geq c
$$
which always holds for large enough $s$.
Now we find a sufficient lower bond for $s$.
Notice that the diagonal dominance property implies $(\underline{p}/\bar{p}) p_i(y\mid y) \geq p_i(x\mid y)$. 
Marginalizing out $x$, we have 
$$
    \left(1 + (|\mathcal{Y}|- 1)\frac{\underline{p}}{\bar{p}}\right) \cdot p_i(y\mid y) \geq 1,
$$
which gives a lower bound on $p_i(y\mid y)$.
Therefore, combining with the above condition, it is sufficient to find $s$ such that
$$
    s\left(\sum_{y} p_Y(y) p_i(y\mid y) - \max p_Y(y) \right)  \geq s\left( \frac{\bar{p}}{\bar{p} + (|\mathcal{Y}|-1)\underline{p}} - \bar{p}\right) \geq c
$$
and solving the second inequality gives a sufficient lower bound.
}

$\square$

\subsection{Proof of Theorem \ref{thm: dram upper bound}}

\textbf{Estimation of the reference distribution $p^*(x_j \mid x_i).$}
As suggested in Section \ref{sec: robust mechanism design}, optimizing the cost of mechanism relies on accurate knowledge over the reference distribution $p^*(x_j \mid x_i)$.
Therefore, we first focus on accurate estimations on this distribution using agents reports.
Throughout this part, agents' reports are assumed to be truthful, i.e. $z_i = x_i$, so intuitively speaking principal should faithfully recover $p*$ if given enough data.

We begin with a lemma on concentration bound on using the empirical estimator for a discrete distribution.
Let $q$ be a discrete distribution on sample space $\mathcal{Y}$, from which we obtain $t$ i.i.d. samples.
Let $\hat{q}$ be the empirical probability distribution where $\hat{q}_t(y) = t_y / t$.
Here $t_y$ is number of times label $y$ appears in the $t$ samples.
We also define $d = |\mathcal{Y}|$.

\begin{lemma}[Concentration inequality of the empirical distribution \cite{weissman2003inequalities}]
    \label{lem: l1 bound on empirical distribution}
    For all $\eta > 0$, we have
    \[
        \mathbb{P}(\tv({q, \hat{q}_t}) \geq \eta) \leq (2^d - 2)\exp\left(-t\varphi(\pi_q)\eta^2\right) \leq (2^d - 2)\exp\left(-2t\eta^2\right),
    \]
    where $\varphi(x) = \log((1-x)/x) / (1-2x)$ with $\varphi(1/2) = 2$, and $\pi_q = \max_{A\subseteq \mathcal{Y}} \min (\mathbb{P}(A)), 1-\mathbb{P}(A))$.
\end{lemma}

Lemma \ref{lem: l1 bound on empirical distribution} gives an concentration inequality on the empirical distribution. 
We now apply this lemma to derive a concentration bound on estimating \textit{conditional} distribution using the empirical \textit{conditional} distribution estimator.
The empirical \textit{conditional} distribution, defined as $\hat{p}_t(x_j \mid x_i) = t_{x_j\mid x_i} / t_{x_i}$, is what we ended up using in Algorithm \ref{alg: dram}.

\begin{lemma}[Concentration property of the empirical conditional distribution]
    \label{lem: bound on conditional distribution}
    Suppose the principal has received $T$ rounds of reports from agent $i$ and $j$.
    Assuming agents are always truthful.
    Let $\hat{p}(x_j \mid x_i)$ be the empirical conditional distribution defined in Algorithm \ref{alg: dram}.
    Define the ambiguity set with ambiguity level $\eta$ as
    \begin{align*}
        S_\eta(\hat{p}) = \big\{ p\in\mathcal{P} \; \big| \; \tv\left(\hat{p}(\cdot \mid x_i),p (\cdot \mid x_i)\right)  \leq \eta,\;\forall x_i\in \mathcal{Y} \cup \{\varnothing\} \big\}.
    \end{align*}
    
    If the number of rounds satisfies
    \begin{align*}
        T \geq \frac{1+2\eta^2}{2\rho\eta^2} \log\left(\frac{(d+1)2^d}{\varepsilon}\right),
    \end{align*}
    where $\rho = \min_{x\in \mathcal{Y}} \mathbb{P}(X_i = x)$, then with probability at least $1-\varepsilon$, the true distribution $p^*$ belongs to $S_\eta$.    
\end{lemma}

We note the increasing rate of $T$ is on the order of $O(\log(1/\varepsilon)/\eta^2)$ for arbitrary $\varepsilon$ and small $\eta$. 
Even when $\eta$ is large, there is still a threshold $T > O(\log(1/\varepsilon))$ that must be satisfied.
This is because there are two possible ways for the event $S_\eta$ to fail: the first case is when the estimator for a certain conditional distribution is $\eta$-away from the true distribution, and the second case is when a certain symbol $x_i$ never appears in $i$'s report.
To ensure the second case does not happen with probability larger than $\varepsilon$, we need $T$ to be large enough.

\begin{proof}
    We first consider one label $x_i \in \mathcal{Y}$.
    Within $T$ rounds, the count of $x_i$ from agent $i$'s report follows a binomial distribution. Let $\rho = \min_{x\in \mathcal{Y}} \mathbb{P}(X_i = x)$.

    We then have
    \begin{align*}
        \mathbb{P}(\tv(\hat{p}(\cdot \mid x_i), p^*(\cdot\mid x_i)) > \eta) &= \sum_{t=0}^T \mathbb{P}(\tv(\hat{p}(\cdot \mid x_i), p^*(\cdot\mid x_i)) > \eta \mid T_{x_i} = t) \cdot \mathbb{P}(T_{x_i} = t) \\
        &\leq \sum_{t=0}^T 2^d \exp (-2t\eta^2) \cdot \binom{T}{t} \rho^t (1-\rho)^{T-t} \\
        &= 2^d \binom{T}{t} \sum_{t=0}^T \left[\rho \exp(-2\eta^2)\right]^t (1-\rho)^{T-t} \\ 
        &= 2^d \left[\rho \exp(-2\eta^2) + 1-\rho\right]^T.
    \end{align*}

    Utilizing union bound across all $(d+1)$ symbols $x_i \in \mathcal{Y} \cup \{\varnothing\}$ would give us
    \begin{align*}
        \mathbb{P}\left(\exists x_i, \tv(\hat{p}(\cdot \mid x_i), p^*(\cdot\mid x_i) > \eta)\right) \leq (d+1)2^d \left[\rho \exp(-2\eta^2) + 1-\rho\right]^T.
    \end{align*}

    Inverting this inequality gives that when
    \begin{align*}
        T \geq \frac{\log((d+1)2^d/\varepsilon)}{-\log(1-\rho(1-\exp(-2\eta^2)))},
    \end{align*}
    the original event in lemma holds with probability at least $1-\varepsilon$.

    Notice that we have 
    \begin{align*}
        \frac{\log((d+1)2^d/\varepsilon)}{-\log(1-\rho(1-\exp(-2\eta^2)))} \leq \frac{\log((d+1)2^d/\varepsilon)}{\rho(1-\exp (-2\eta^2))} \leq \frac{1+2\eta^2}{2\rho \eta^2} \log\left(\frac{(d+1)2^d}{\varepsilon}\right),
    \end{align*}
    where the first inequality holds since $-\log(1-x)\ge x$, and the second holds since $1-\exp(-x)\ge x/(1+x)$.
    Thus a sufficient bound is
    \[
        T \geq \frac{1+2\eta^2}{2\rho \eta^2} \log\left(\frac{(d+1)2^d}{\varepsilon}\right).
    \]
    Also, for $\eta < 1/\sqrt{2}$, we have a sufficient bound 
    \[
        T \geq \frac{1}{\rho \eta^2} \log\left(\frac{(d+1)2^d}{\varepsilon}\right).
    \]
\end{proof}

\textbf{Warm-starting.}
The ambiguity threshold $\tilde{\eta}$ is the smallest value across agents, on the maximum ambiguity a distributionally robust mechanism can tolerate.
For mathematical convenience we update the parameter to let $\tilde{\eta} < 1/\sqrt{2}$, so that we have a cleaner bound in subsequent derivations.
The fact-checking mechanism ensures truthfulness because of Lemma \ref{lem: warm-start fact-checking}.
The length of this phase is $O(\log (N\log T))$, which results in smaller order total regrets even when we have constant regret in each round.

\textbf{Bounding the regret.}
Now we focus on the algorithm for a single agent $i$.
Consider what happens in epoch $k$, where $t \in (\tau_{k-1},\tau_k]$.
At the beginning of the epoch, we have $\tau_{k-1} = \blue{2^{k-1}}\tau$ data points, and we set $\eta_k$ in a specific way so that
\[
    \tau_{k-1} \geq \frac{1}{\rho\eta_k^2} \log\left(\frac{(d+1)2^d N \blue{m}}{\varepsilon}\right),
\]
thus we know from Lemma \ref{lem: bound on conditional distribution} that $p_i^* \in S_{\eta_k }(\hat{p}_{ik})$ with probability at least $1-\varepsilon/N m$.

From Theorem \ref{thm: cost of robustness}, we know that a mechanism $R_{ik}$ can be constructed that guarantees agent $i$'s truthfulness when $p_i^* \in S_{\eta_k }(\hat{p}_{ik})$.
Also, in the duration of epoch $k$, the expected regret for a single round is $c \cdot C_1\eta_k/(1-C_2\eta_k)$ where $C_1$ and $C_2$ are constants as defined in Eq.(\ref{eq: cost of robustness}).
Therefore, the total expect regret for agent $i$ across all $T$ rounds is
\begin{align*}
    \sum_{k=1}^m \frac{C_1 \eta_k}{1-C_2 \eta_k} c \cdot (\tau_k - \tau_{k-1}) &\leq \frac{C_1 }{1-C_2/\sqrt{2}} c \cdot \sum_{k=1}^m \eta_k \cdot (\tau_k-\tau_{k-
    1}) \\
    &= \frac{C_1 }{1-C_2/\sqrt{2}} c \cdot \sqrt{\log((d+1)2^d N m/\varepsilon)} \cdot\sum_{k=1}^m \frac{1}{\sqrt{\tau_{k-1}}} \cdot (\tau_k-\tau_{k-
    1}) \\ 
    &= \frac{C_1 }{1-C_2/\sqrt{2}} c \cdot \sqrt{\log((d+1)2^d N m/\varepsilon)} \cdot\sum_{k=1}^m \sqrt{2^{k-1}\tau_0} \\
    &\precsim O(\sqrt{T}\log(d \cdot 2^d N m /\varepsilon))  
\end{align*}
In the last inequality we used the equation on the sum of geometric sequences.

Finally, we know that for one agent in one epoch, the scheme ensures that applying union bound across all $N$ agents and $m$ periods.
Hence applying union bound, we know that truthfulness is held with probability $1-\varepsilon$. 
Also, note that the regret in the warm start phase is at most $O(\tau)$ and therefore dominated by the regret from the adaptive phase.
Thus the total regret is $N$ times the single agent's regret and therefore $O(N\sqrt{T}\log(N m/\varepsilon)).$
With the fact that $m \leq \lceil\log_2T\rceil.$ we have the desired bound.

\subsection{Proof of Corollary \ref{col: dram upper bound known t}}

The updated upper bound under the epoch schedule $\tau_k - \tau_{k-1} = T^{1-2^{-(k-1)}}\tau$.
All steps are identical to the proof of Theorem \ref{thm: dram upper bound}, except for the final step, where we sum over the regret across epochs:
\blue{\begin{align*}
\sum_{k=1}^m \frac{C_1 \eta_k}{1-C_2 \eta_k} \, c \, (\tau_k - \tau_{k-1})
&\leq \frac{C_1}{1-C_2/\sqrt{2}} \, c \sum_{k=1}^m \eta_k (\tau_k - \tau_{k-1}).
\end{align*}
Using $\eta_k = \sqrt{\log((d+1)2^d N m/\varepsilon)} / \sqrt{\tau_{k-1}}$, it suffices to bound
\[
\sum_{k=1}^m \frac{\tau_k - \tau_{k-1}}{\sqrt{\tau_{k-1}}}.
\]
Let $L_k := \tau_k - \tau_{k-1}$. Then $L_1=\tau$ and $\tau_0=\tau$, so $L_1/\sqrt{\tau_0}=\sqrt{\tau}$. For $k\ge2$, since $\tau_{k-1}\ge L_{k-1}$,
\[
\frac{\tau_k - \tau_{k-1}}{\sqrt{\tau_{k-1}}}
\le \frac{L_k}{\sqrt{L_{k-1}}}
= \sqrt{T}\,\sqrt{\tau_0}.
\]
(The same bound holds if the last epoch is truncated.) Hence,
\[
\sum_{k=1}^m \frac{\tau_k - \tau_{k-1}}{\sqrt{\tau_{k-1}}}
\le m \sqrt{T}\sqrt{\tau_0}.
\]
Therefore,
\[
\sum_{k=1}^m \eta_k (\tau_k - \tau_{k-1})
= O\!\left(
m \sqrt{T}\sqrt{\tau}
\sqrt{\log\!\left(\frac{(d+1)2^d N m}{\varepsilon}\right)}
\right).
\]
Choosing $\tau_0 = \Theta\!\big(\log\!\big(\tfrac{(d+1)2^d N m}{\varepsilon}\big)\big)$ yields
\[
O\!\left(
m \sqrt{T}
\log\!\left(\frac{(d+1)2^d N m}{\varepsilon}\right)
\right).
\]
Summing over $N$ agents gives the final bound.
}

We note that this epoch schedule is sub-geometric, but it grows faster at the first few steps than the geometric epoch schedule $\tau_k - \tau_{k-1} =2^{k-1}\tau$, therefore it uses logarithmically smaller $O(\log\log T)$ epoch count to reach $T$.

\subsection{Proof of the Lower Bound (Theorem \ref{thm: lower bound})}

We start by considering the two-agent, two-label case.
Let the label space be $\mathcal{Y} = \{0, 1\}$.
Without loss of generality we let cost $c=1$.

We first construct a pair of problem instances $p_0, p_1 \in \Delta(\mathcal{Y}^N)$ that simultaneously satisfies the two requirements: i) the two instances are statistically close; ii) the optimal mechanisms differ sharply.
Let the focal agent be agent $1$ and the reference agent be $2$.
We now study the reward mechanisms for the focal agent $1$.
Consider the following examples on $p_0$ and $p_1$, whereas from the focal agent's perspective, the two instances has:
\begin{align*}
    \mathbf{B}_0 = \begin{bmatrix}
        0.5 - \delta & 0.5+\delta\\
        1 & 0
    \end{bmatrix}&,
    \quad \mathbf{B}_1 = \begin{bmatrix}
        0.5 + \delta & 0.5-\delta\\
        1 & 0
    \end{bmatrix}\\
    \mathbb{P}_k(X_1 = x) = 0.5&,  \quad \forall x\in \mathcal{Y},\; k\in \{0,1\}.
\end{align*}
We call $\delta$ the \textit{cheapness} parameter, as it relates to how low the principal's expected payment can be.

\begin{lemma}[Hard instances]
    \label{lem: lower bound competition and similarity}
    For the aforementioned two instances $p_0$ and $p_1$ with parameter $\delta \in (0,1/4)$, we have:
    \begin{itemize}
        \item Competition: under instance $p_0$, any reward mechanism $R$ is feasible with respect to constraints in Eq.\eqref{eq: LP formation for two-agent mechanism design} and whose expected payment \blue{under truthful reporting} is less than $1+\delta$ must either violate IC constraints or pays more than $1+\delta$ under instance $p_1$.
        The statement holds with the roles of $p_0$ and $p_1$ reversed.

        \item Similarity: $p_0$ and $p_1$ are statistically close, i.e., $\operatorname{KL}(p_0 \| p_1) \leq 8 \delta^2$.
    \end{itemize}
    \blue{Finally, instances $p_0$ and $p_1$ are not pathological since corresponding $\mathbf{B}_0$ and $\mathbf{B}_1$ are bounded away from singular matrices.}
\end{lemma}

\begin{proof}
    \textit{Competition}.
    We follow a similar procedure as in Section \ref{appendix: proof of optimal cost of 2-agent peer prediction}.
    Let $\mathbf{B}\mathbf{R}^\intercal = \mathbf{M}$, notice that $\mathbf{M}$ and $\mathbf{R}$ has 1-1 correspondence since both $\mathbf{B}_k$ are invertible.
    We therefore studies in the space of $\mathbf{M}$.
    The problem becomes:
    \begin{align*}
        \operatorname{Find} &\quad \mathbf{M} = \begin{bmatrix}
            a & b \\
            c & d \\
        \end{bmatrix} \\
        \operatorname{s.t.} &\quad \mathbf{M}_{xx} \geq 1, \quad \forall x\in \mathcal{Y} \\
        &\quad \mathbf{M}_{xy} \leq 1, \quad \forall x\neq y\in \mathcal{Y}\\
        &\quad \mathbf{M}^\intercal \mathbf{d}' \leq \mathbf 0 \\
        &\quad \sum_x \mathbb{P}(X_1 = x) \mathbf{M}_{xx} \leq 1 + \delta
    \end{align*}
    where $\mathbf{d}'_x = \mathbb{P}(X_1 = x) = 1/2$.
    This gives a set of necessary conditions on the entries for feasible $\mathbf{M}$:
    \begin{align*}
        1 \leq a & \leq 1+2\delta \\
        1 \leq d & \leq 1+2\delta \\
        b & \leq 1 \\
        c & \leq 1 \\
        a + c & \leq 0 \\
        b + d & \leq 0.
    \end{align*}
    (Here the $1+2\delta$ is a relaxation on the bound $(a+d)/2 < 1+\delta$.)
    Notice that $\mathbf{M}_{xy}$ is exactly agent $1$'s expected reward given she observes label $x$ and reports label $y$.
    Therefore, for any $\mathbf{M}_0$ that is cheap and satisfies the constraints under $p_0$, its performance under $p_1$ is $\mathbf{M}_1 = \mathbf{B}_1 \mathbf{B}_0^{-1} \mathbf{M}_0$.
    However, we have that:
    \begin{align*}
        \mathbf{M}_1 = \begin{bmatrix}
            \frac{1-2\delta}{1+2\delta} & \frac{4\delta}{1+2\delta} \\
            0 & 1
        \end{bmatrix}
        \begin{bmatrix}
            a & b \\
            c & d
        \end{bmatrix}
        = \begin{bmatrix}
            \frac{(1-2\delta) a+ 4\delta c}{1+2\delta}  &  \frac{(1-2\delta) b+ 4\delta d}{1+2\delta}\\
            c & d \\
        \end{bmatrix}.
    \end{align*}
    Therefore, the first entry must have
    \begin{align*}
        \frac{(1-2\delta) a+ 4\delta c}{1+2\delta} \leq  \frac{(1-2\delta) (1+2\delta) - 4\delta}{1+2\delta} < 1,
    \end{align*}
    leading to a violation of the truthfulness constraint under $p_1$.
    Similar procedure would also prove the statement with $p_0$ and $p_1$ reversed.
    For any cheap and feasible mechanism under $p_1$, the resulting first entry of corresponding $\mathbf{M}_0$ would be
    \begin{align*}
        \frac{(1+2\delta)a - 4\delta c}{1-2\delta} \geq \frac{(1+2\delta) + 4\delta }{1-2\delta} > 1+\delta.
    \end{align*}
    while the last entry is $d \geq 1$.
    This implies that such mechanism would violate the cheapness constraint under $p_0$.
    Hence we prove the competition property in both ways.
    
    \textit{Similarity}.
    \begin{align*}
        \operatorname{KL}(p_0 \| p_1) &= \frac{1-2\delta}{4}\log\frac{(1-2\delta)/4}{(1+2\delta)/4} + \frac{1+2\delta}{4}\log\frac{(1+2\delta)/4}{(1-2\delta)/4} + \frac{1}{2}\log \frac{1/2}{1/2} \\ 
        &= \delta \log\left(1 + \frac{4\delta}{1-2\delta}\right) \\
        &\leq \frac{4\delta^2}{1-2\delta} \\
        &\leq 8\delta^2.
    \end{align*}
    The first inequality holds since $\log(1+x) \leq x$, and the second holds since $\delta \in (0, 1/4)$.
\end{proof}

\blue{
Now we consider the sequential mechanism design problem.
Lemma \ref{lem: lower bound competition and similarity} gives a pair of hard instances in the two-agent, two-label case.
We would provide a direct-sum construction to create a hard instance for the $N$-player, two-label case.
Apparently, this instance would still apply even when labels are more than two.}

\blue{Group the $N$ players into $M = \lfloor N/2 \rfloor$ pairs where $j = \{1, 2,\dotsb, M\}$ denotes the pair index and agent $i = 2j -1$ and $2j$ belongs to pair $j$.
Fix $\theta \in \{0, 1\}^M$.
Define a pair-independent instance $p_\theta \in \Delta(\mathcal{Y}^N)$ by
$$
    p_\theta(\boldsymbol{x}) = \prod_{j=1}^M p_{\theta_j} (x_{2j-1}, x_{2j}) \cdot \prod_{i = 2M+1}^N \mathbf{1}\{x_i = 0\}.
$$
Essentially, in this instance each pair independently selects one of the hard instances $p_0$ or $p_1$, and their selections are parameterized by $\theta$.}

\blue{Let $\mathcal{F}_{t-1}$ denote the $\sigma$-field generated by all the history up to the end of round $t-1$, including randomness from the policy.
In round $t$, the policy announces a reward mechanism $R_{it}: \{0, 1\}^N \to \mathbb{R}$ for agent $i$, based on the available history from $\mathcal{F}_{t-1}$.
Write $i = 2j-1$.
Define the two-agent projected mechanism under $p_\theta$ by:
$$
\tilde{R}_{it}(z_i, z_{i+1}) = \mathbb{E}
[R_{it} (z_i, z_{i+1}, \boldsymbol{X}_{-(i, i+1), t}) \mid \mathcal{F}_{t-1}]$$
This projection specifies the payment for the agent $i$, averaged over the randomness caused by cross-pair influence from other pairs.
Since each pair $j$ is independent from others, we have $\boldsymbol{X}_{-(i, i+1)} \indep X_i, X_{i+1}$.
It follows that for each $x, z \in \{0,1\}$, we have
\begin{align*}
    \mathbb{E}[R_{it}(z, \boldsymbol{X}_{-i}) \mid X_i = x, \mathcal{F}_{t-1}] &= \mathbb{E}[\mathbb{E}[R_{it}(z,  X_{i+1}, \boldsymbol{X}_{-(i, i+1)})  \mid X_{i+1}, X_i = x, \mathcal{F}_{t-1}]\mid X_i = x, \mathcal{F}_{t-1}] \\
    &= \mathbb{E}[\tilde{R}_{it}(z, X_{i+1}) \mid X_i = x, \mathcal{F}_{t-1}].
\end{align*}
This indicates that from agent $i$'s perspective, any $N$-agent is effectively equivalent to the two-agent projection to her belonged pair.
This is because knowing $X_i$ only changes the belief of the other agent in the pairs but not the rest.
Likewise, this reduction holds for unconditional expectations under any report $z$.
}

\blue{Now for $k\in \{0, 1\}$, we define the good event $G_k$ where
$$
    G_k = \{R: R \text{ is feasible for Eq.\eqref{eq: two-agent mechanism design} under $p_k$ and } \mathbb{E}_{p_k}[R(X_1, X_2)] \leq 1+\delta\}.
$$
Lemma \ref{lem: lower bound competition and similarity} proves that $G_0 \cap G_1 = \varnothing$.
For a given $\theta$ and an agent $i = 2j-1$, define the event 
$$
    E_{ik} = \{\tilde{R}_{it} \text{ is feasible under $p_k$ for all } t\} \cap \left\{\sum_{t=1}^T \left(\mathbb{E}_{p_k}[\tilde{R}_{it}(X_1, X_2)] - 1 \right) \leq \delta T /2\right\}.
$$
On $E_k$, any round $t$ with $\tilde{R}_{it} \notin G_k$ must have $\mathbb{E}_{p_k}[\tilde{R}_{it}(X_1, X_2)] > 1+\delta$.
This is because feasibility holds on $E_k$, so the only way is to violate the cheapness condition in $G_k$.
This indicates that on $E_k$, we must have
$$
    \# \{t \in \{1, \dotsb, T\}: \tilde{R}_{it} \in G_k \} > T/2,
$$
since the total regret is smaller than $\delta T/2$.
}

\blue{Now, for that same pair $j$, fix the rest pairs' choices $\theta_{-j}$, and also fix all the players' reporting strategies up to round $T$.
Define the estimator $\hat{\theta}_j$ which outputs $0$ if $\# \{t \in \{1, \dotsb, T\}: \tilde{R}_{it} \in G_0 \} > T/2$, and outputs $1$ otherwise.
Notice that given $\theta_{-j}$, one can exactly compute $\tilde{R}_{it}$ using info before $t$ since $p_{\boldsymbol{X}_{-(i, i+1)}}$ is available.
Therefore $\hat{\theta}_j$ is indeed a valid test function on $\mathcal{F}_{t-1}$.
Since $G_0 \cap G_1 =\varnothing$, we have
$$
    E_{i0} \subseteq \{\hat{\theta}_j = 0\} \quad \text{ and } \quad E_{i1} \subseteq \{\hat{\theta}_j = 1\}. 
$$
Hence
$$
    \mathbb{P}_0 (E_{i0}^c) \geq \mathbb{P}_0(\hat{\theta}_j = 1)  \quad \text{  and } \quad \mathbb{P}_1 (E_{i1}^c) \geq \mathbb{P}_1 (\hat{\theta}_j = 0).
$$
Let $\mathcal{T}_{jk}$ denotes the full transcripts (all public data, including reports and possibly random reward but not private observations) up to $T$, generated according to the law whereas $\theta_j = k \in \{0, 1\}$ while holding the rest law ($\theta_{-j}$, strategies, and reward randomness) stays fixed.
Consider $p_{(0, \theta_{-j})}$ and $p_{(1, \theta_{-j})}$.
Since they differ only on the independent block $j$, we have 
$$
    \operatorname{KL}(p_{(0, \theta_{-j})} \| p_{(1, \theta_{-j})}) = \operatorname{KL}(p_0 \| p_1)\leq 8\delta^2.
$$
Over $T$ i.i.d. rounds, tensorization gives $\operatorname{KL}(p^T_{(0, \theta_{-j})} \| p^T_{(1, \theta_{-j})}) \leq 8T\delta^2$.
By data processing, we have $\operatorname{KL}(\mathcal{T}_{j0} \| \mathcal{T}_{j1}) \leq 8T\delta^2$. 
Therefore applying Bretagnolle-Huber, one has
$$
    \mathbb{P}_0(\hat{\theta}_j = 1) + \mathbb{P}_1 (\hat{\theta}_j = 0) \geq \frac{1}{2}\exp(-\operatorname{KL}(\mathcal{T}_{j0} \| \mathcal{T}_{j1})) \geq \frac{1}{2}\exp(-8T\delta^2).
$$
Since this holds for arbitrary $\theta_j$, averaging over all possible $2^M$ $\theta$ realizations, and summing over blocks $j = 1,\dotsb, M$ gives us
$$
\frac{1}{2^M} \sum_\theta \mathbb{P}_\theta(E_{i\theta}^c) \geq \frac{1}{2^M}  \sum_\theta \mathbb{P}_\theta (\hat{\theta}_j \neq \theta_j) \geq \frac{1}{4} \exp(-8T\delta^2).
$$
}

\blue{Now consider a policy such that for arbitrary instance, with probability at least $1-\varepsilon$, the mechanism is feasible for all players across all $T$ rounds.
This means for each player $i = 2j - 1$, feasibility failure contributes to $\mathbb{P}_\theta(E_{i\theta}^c)$ at most $\varepsilon$ under every possible $\theta$.
If $\mathbb{P}_\theta(E_{i\theta}^c) > \varepsilon$, the rest probability mass must come from the event that
$$
    \left\{\sum_{t=1}^T \left(\mathbb{E}_{p_\theta}[\tilde{R}_{it}(X_1, X_2)] - 1 \right) > \delta T /2\right\}.
$$
}

Now, the regret for player $i$ is $\reg_i = \sum_{t} (R_{it} - 1)$.
Under truthful reporting, this mechanism has
\begin{align*}
    \mathbb{E}[\reg_i] &= \mathbb{E}\left[ \sum_{t}(R_{it} -1)\right] \\
    &= \mathbb{E}\left[ \sum_{t}\mathbb{E}_{p_\theta}\left[(R_{it}(\boldsymbol{X}_t) -1) \mid \mathcal{F}_{t-1}\right]\right] \\
    &= \mathbb{E}\left[ \sum_{t}\mathbb{E}_{p_{\theta_j}}\left[(\tilde{R}_{it}(X_1, X_2) -1) \mid \mathcal{F}_{t-1}\right]\right].
\end{align*}
Combining this fact with the lower bound, and summing the regret over $j= 1,\dotsb, M$, we have that 
$$
    \frac{1}{2^M} \sum_\theta \sum_{j}\mathbb{E}_\theta [\reg_{2j-1}] \geq M \cdot \frac{\delta T}{2} \cdot \left(\frac{1}{4}\exp(-8T\delta^2) -\varepsilon \right).
$$
Now, choose a proper $\delta$:
$$
    \delta = \sqrt{\frac{1}{8T}\log \left(\frac{2}{1+4\varepsilon}\right)} \in (0, 1/4),
$$
and with that we have $\exp(-8T\delta^2) = (1+4\varepsilon)/2$, and the bracket equals $(1-4\varepsilon) / 8$. 
Therefore, under truthful reporting there exists some $\theta^*$ such that
$$
    \sum_{j}\mathbb{E}_{\theta^*} [\reg_{2j-1}] \geq M \cdot \frac{1-4\varepsilon}{16} \cdot \sqrt{\frac{T}{8} \log \left(\frac{2}{1+4\varepsilon}\right)}.
$$
Finally, when the policy is feasible for all players across all rounds, we have the individual rationality constraint that leads to $\mathbb{E}[R_{it}(\boldsymbol{X}_t)] \geq 1$, thus for the rest of the players regret is greater than $0$.
Therefore, under the event of feasibility for all players across all rounds and with $M = \lfloor N/2\rfloor \geq (N-1)/2$, we have
$$
    \mathbb{E}_\theta[\reg \mid \text{Feasibility}] \geq \sum_{j}\mathbb{E}_{\theta^*} [\reg_{2j-1}] \geq \frac{1-4\varepsilon}{64\sqrt{2}} (N-1) \sqrt{T \log \left(\frac{2}{1+4\varepsilon}\right)}.
$$

$\square$

\subsection{Proof of Theorem \ref{thm: dram-plus upper bound}}

The proof roughly follows the same procedure as Theorem \ref{thm: dram upper bound}, with a few modifications.
First, since we don't have an explicit formula for the PAC guarantee, we cannot invert the function $\eta_\varepsilon(T)$ to get a closed-form bound for $T$ under certain $\eta$ and $\varepsilon$.
This may lead to a relatively looser bound for certain estimators.
The tightest bound can always be specifically derived following proof of Theorem \ref{thm: dram upper bound}.
Second, we would use the estimator for conditional distribution $p(\cdot \mid x_i)$, for each $x_i \in \mathcal{Y}$.
For each $x \in \mathcal{Y}$, define the random sample size $T_x$ as the number of times $x$ is observed before and including round $T$.
There are two scenarios where the estimator could be off: 
\begin{enumerate}
    \item agent $i$ does not observe $x_i$ for enough number of times. ($T_{x_i}$ is small).
    \item The estimation on $p(\cdot \mid x_i)$ is off.
\end{enumerate}
The first scenario is not decided by whatever the estimator used by principal, since it is a tail events of a multinomial distribution.
For the same reason, we cannot directly merge the two probabilities together as it is done in Lemma \ref{lem: bound on conditional distribution}, resulting in the following lemma.

\begin{lemma}[Concentration property of general discrete distribution estimator]
    \label{lem: bound for general estimator}
    Suppose the principal has received $T$ rounds of reports from agent $i$ and $j$.
    Assuming agents are always truthful.
    Let $\hat{p}(x_j\mid x_i)$ be the conditional distribution estimation from the general estimator in Definition \ref{def: general estimator}.
    \blue{Let $p^*$ be the true sampling distribution with $\rho = \min_{x\in \mathcal{Y}} p^*(X_i = x)$.}
    Define the ambiguity set with ambiguity level $\eta_\varepsilon(\rho T/2)$ as
    \begin{align*}
        S_{\eta_{\varepsilon}(\rho T/2)}(\hat{p}) = \big\{ p\in\mathcal{P} \; \big| \; \tv\left(\hat{p}(\cdot \mid x_i),p (\cdot \mid x_i)\right)  \leq \eta_\varepsilon(\rho^*( T/2),\;\forall x_i\in \mathcal{Y} \cup \{\varnothing\} \big\}.
    \end{align*}
    Then, with probability at least $1-(d+1)(\varepsilon + \exp(-\rho T/8))$, the true distribution $p^*$ belongs to $S_{\eta_{\varepsilon}(\rho T/2)}$.
\end{lemma}

\begin{proof}
    For one label $x_i \in\mathcal{Y}$, we have
    \begin{align*}
        \mathbb{P}(\tv (\hat{p}(\cdot \mid x_i), p^*(\cdot \mid x_i)) > \eta) &= \sum_{t=0}^T \mathbb{P}(\tv (\hat{p}(\cdot \mid x_i), p^*(\cdot \mid x_i)) > \eta \mid T_{x_i} = t) \cdot \mathbb{P}(T_{x_i} = t) \\ 
        &\leq \mathbb{P}(T_{x_i} \leq \rho T/ 2) + \blue{\mathbb{P}(\tv (\hat{p}(\cdot \mid x_i), p^*(\cdot \mid x_i)) > \eta \mid T_{x_i} > \rho T/2)} \\
        &\leq \exp(-\rho T/8) + \varepsilon
    \end{align*}
    \blue{The first inequality is when we split the failure event into the aforementioned two scenarios, and the second inequality is Chernoff bound applied to the binomial distribution $\operatorname{Bin}(T, \rho)$}.

    Applying union bound across all labels would give us
    \begin{align*}
        \mathbb{P}(\exists x_i, \tv (\hat{p}(\cdot \mid x_i), p^*(\cdot \mid x_i)) > \eta) \leq (d+1)(\exp(-\rho T/8) + \varepsilon).
    \end{align*}
    
\end{proof}

Now we focus on the regret for a single agent $i$.
Consider what happens in epoch $k$, where $t \in (\tau_{k-1},\tau_k]$.
At the beginning of the epoch, we have $\tau_{k-1}$ data points, and we set $\eta_k = \eta_{\varepsilon/Nm(d+1)}(\rho\tau_{k-1}/2)$.
Thus we know from Lemma \ref{lem: bound for general estimator} that $p_i^* \in S_{\eta_k }(\hat{p}_{ik})$ with probability at least $1-\varepsilon/(Nm) - (d+1)\exp (-\rho \tau_{k-1}/8) $.

From Theorem \ref{thm: cost of robustness}, we know that a mechanism $R_{ik}$ can be constructed that guarantees agent $i$'s truthfulness when $p_i^* \in S_{\eta_k }(\hat{p}_{ik})$.
Also, in the duration of epoch $k$, the expected regret for a single round is $c \cdot C_1\eta_k/(1-C_2\eta_k)$ where $C_1$ and $C_2$ are constants as defined in Eq.(\ref{eq: cost of robustness}).
Therefore, the total expect regret for agent $i$ across all $T$ rounds is
\begin{align*}
    \sum_{k=1}^m \frac{C_1 \eta_k}{1-C_2 \eta_k} c \cdot (\tau_k - \tau_{k-1}) &\leq C_1 c \cdot \sum_{k=1}^m \eta_k \cdot (\tau_k-\tau_{k-
    1}) \\
    &= C_1 c \cdot \sum_{k=1}^m \eta_{\varepsilon/Nm(d+1)} (\rho \tau_{k-1}/2) \cdot (\tau_k - \tau_{k-1})
\end{align*}

Finally, we know that for one agent in one epoch, the scheme ensures that applying union bound across all $N$ agents and $m$ periods. 
Hence applying union bound, we know that truthfulness is held with probability 
\[
    1 - \varepsilon - N(d+1) \cdot \sum_{k=1}^m \exp(-\rho \tau_{k-1}/8).
\]
Also, note that the regret in the warm start phase is at most $O(\tau)$ and therefore dominated by the regret from the adaptive phase.
Thus the total regret is $N$ times the single agent's regret and therefore 
\[
    O\left(N\sum_{k=1}^m \eta_{\varepsilon/Nm(d+1)} (\rho \tau_{k-1}/2) \cdot (\tau_k - \tau_{k-1})\right)
\]
$\square$

\end{document}